\documentclass[%
 aip,
 amsmath,amssymb,
 reprint,%
onecolumn,
]{revtex4-1}
\usepackage{color}
\usepackage{cancel}

\usepackage{graphicx}
\usepackage{dcolumn}
\usepackage{bm}

\usepackage{lineno}

\usepackage[utf8]{inputenc}
\usepackage[T1]{fontenc}
\usepackage{mathptmx}
\usepackage{amsthm}

\newcommand{\Ron}{R_{\text{on}}}                     
\newcommand{\Roff}{R_{\text{off}}}   

\newcommand{\comm}[1]{}

\newcommand{\colvec}[2][.8]{%
  \scalebox{#1}{%
    \renewcommand{\arraystretch}{.8}%
    $\begin{pmatrix}#2\end{pmatrix}$%
  }
}
\newtheorem{theorem}{Theorem}

\begin{document}

\preprint{AIP/123-QED}

\title[Memristive circuits]{Network analysis of memristive device circuits: dynamics, stability and correlations}

\author{F. Barrows}
 \affiliation{Theoretical Division (T4), Los Alamos National Laboratory, Los Alamos, New Mexico 87545, USA}
  \affiliation{Center for Nonlinear Studies, Los Alamos National Laboratory, Los Alamos, New Mexico 87545, USA}
\author{F. Sheldon}
 \affiliation{London Institute for Mathematical Sciences, Royal Institution, 21 Albemarle St, London W1S 4BS, UK}
\author{F. Caravelli}
 \affiliation{Theoretical Division (T4), Los Alamos National Laboratory, Los Alamos, New Mexico 87545, USA}


\begin{abstract}

Networks with memristive devices are a potential basis for the next generation of computing devices.
They are also an important model system for basic science, from modeling nanoscale conductivity to providing insight into the information-processing of neurons.
The resistance in a memristive device depends on the history of the applied bias and thus displays a type of memory.
The interplay of this memory with the dynamic properties of the network can give rise to new behavior, offering many fascinating theoretical challenges.
But methods to analyze general memristive circuits are not well described in the literature.
In this paper we develop a general circuit analysis for networks that combine memristive devices alongside resistors, capacitors and inductors and under various types of control.
We derive equations of motion for the memory parameters of these circuits and describe the conditions for which a network should display properties characteristic of a resonator system.
For the case of a purely memresistive network, we derive Lyapunov functions, which can be used to study the stability of the network dynamics.
Surprisingly, analysis of the Lyapunov functions show that these circuits do not always have a stable equilibrium in the case of nonlinear resistance and window functions.
The Lyapunov function allows us to study circuit invariances, wherein different circuits give rise to similar equations of motion, which manifest through a gauge freedom and node permutations.
Finally, we  identify the relation between the graph Laplacian and the operators governing the dynamics of memristor networks operators, and we use these tools to study the correlations between distant memristive devices through the effective resistance.
\end{abstract}

\maketitle

\tableofcontents

\section{Introduction}

Once regarded as somewhat niche, memristance is now recognized as a ubiquitous feature of nature, especially at the nanoscale~\cite{Pershin2011}.
As a result, memristive devices have been the subject of intense research over the past decade. In addition to the wide applicability of these components, ranging from scalable memory devices to machine learning and reservoir computing\cite{DCRAM,Carbajal2015,reservoirmem,milano001,zhurc,Xu_FNano_2021}, there is also interest in their fundamental dynamical properties\cite{Pershin2013,Diaz-Alvarez2019,Zhu2021information,caravelli2017complex,Riaza_NonLinAnalysis_2011}. For example, it has long been known that the introduction of memristive devices in a circuit can lead to chaotic dynamics and there are numerous papers studying the properties of the Chua attractors\cite{Matsumoto_IEEE_1984,Maden_ChuaCircuit_1997}.

However, less is known for generic circuits with memory, and even fewer exact results have been obtained in this regard. While little theoretical progress has been made for the most general cases of electronic circuits with memory, purely memristive circuits are amenable to an analytical approach: the equations for networks of these devices can often be obtained analytically, unveiling the most bizarre and unexpected properties of these dynamical circuits\cite{Zegarac_2019,TelliniMacucci_CircuitTheory_2021}. For instance, it has been shown that the asymptotic states of certain memristive-resistive circuits can be well described by the Curie-Wei\ss\ model and its ferromagnetic-paramagnetic transition. Further, it has been argued that random circuits of memristive devices (analyzed through the lens of their Lyapunov function) exhibit a ferromagnetic-glass transition\cite{caravellisheldon,Caravelli4}.

Obtaining a full dynamic picture of this behavior is difficult because the nonlinearity of the memristive devices makes the equations of motion impossible to solve analytically. This paper provides tools to understand the properties of generic, first-order memristive devices with window functions in networks. In particular, we derive dynamical equations for generic circuits with memristive devices, and we derive Lyapunov functions for the infinite family of first-order memristive circuits with sharp boundaries (e.g., discontinuous window functions).

We also generalize techniques used for network circuit analysis to the case of memristive device networks. Using these, we demonstrate that these circuits frequently give rise to analytically nontrivial or intractable dynamical equations. In doing so, we demonstrate mathematical techniques, including correlation analysis and Lyapunov functions, which can provide information on the dynamics, evolution, and stability of a network of memristive devices. Further, we discuss how different circuits give rise to similar dynamics through gauge freedom and node permutation in the governing equations.

This work builds upon previous works to study the dynamics of memristor networks most of which use the graph Laplacian or projector operators\cite{caravelli2017complex,Stern_arxiv_2024,Xiang_Chaos_2022,Dhivarkaran_NeurProcLett_2024}. We identify the relation between the graph Laplacian and the projection operators. Finally we use this correspondence to study the effective resistance in memristor networks and construct effective circuits which allow us to study the correlations between distant memristors.

\section{Graph theory for circuit dynamics}

In network analysis of electrical circuits, an electrical circuit is mapped to a graph $\mathcal{G}=(V, E)$. A graph $\mathcal{G}$ consists of a set of $n$ vertices/nodes $V$ and a set of $m$ edges $E$. An undirected edge may be denoted $\{k, l\}$, but in circuit analysis, each edge is assigned an arbitrary direction from $k$ to $l$ denoted $(k,l)$. The graph allows no self-edges and no duplicate edges. As a simple example, let us work with the triangle graph defined by,
\begin{align}
V = \{1, 2, 3\}, \quad E = \{(1,2), (2,3), (1,3)\}
\end{align}
A representation of this graph is shown in Figure \ref{fig:SimpleTriangle}.
\begin{figure}[ht]
\includegraphics[width=.2\textwidth]{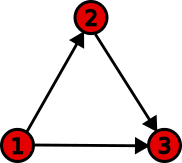}
\caption{ A simple graph with three edges connecting three nodes without an oriented cycle.}
\label{fig:SimpleTriangle}
\end{figure}
Note that the edge between 1 and 3 is not cyclical. A graph representation of a circuit generally consists of edges corresponding to circuit elements and nodes, i.e., junctions in the circuit. Currents and voltages are assigned in the graph to satisfy Kirchhoff's current law (KCL) and Kirchhoff's voltage law (KVL). A circuit is solved when it satisfies these constraints given a circuit topology; thus, network analysis of a circuit consists of solving for valid current and voltage configurations.  

A current configuration $i = \lbrace i_e, \; e\in E\rbrace$ is a set of numbers associated with the edges which satisfy KCL. KCL could be enforced on each vertex by building an $n\times m$ incidence matrix $B$, where each row represents a node and each column an edge; the matrix entries take values of $\pm1$ corresponding to a directed edge oriented towards or away from each node, and $0$ when an edge is not incident on a node.
For our triangle graph, this is,
\begin{equation}
B = \bordermatrix{
  & e_{12} & e_{23} & e_{13} \cr
  n_1 &1  &  0 &  1 \cr
n_2 &-1 &  1 &  0 \cr
n_3 &0  & -1 & -1 \cr
}
\end{equation}
$e_{ij}$ are edges linking nodes $i$ and $j$ (edges are indexed with superscripts, $e^k$), and $n_i$ are indexed nodes. For each vertex $k$ and edge $e$, let $b_{k,e} = 1$ if $k$ is the start, $-1$ if $k$ is the finish, and $0$ if the edge does not include $k$. KCL are enforced such that a current vector $\vec{i}$ satisfies,
\begin{eqnarray}
    \sum_e b_{k,e}i_e = 0 .
\end{eqnarray}
If $\vec{i}$ is an allowed current configuration then
\begin{eqnarray}
    B \vec{i} = 0.
\end{eqnarray}
The operator $B$ may be interpreted as a type of divergence on the graph, and the relationship $B\vec{i} = 0$ is equivalent to saying that the current configurations must be divergence-less. In this sense, they are similar to magnetic fields.

Similarly, a voltage configuration $v =\lbrace v_e, \; e\in E\rbrace$ is a set of numbers associated with the edges that satisfy KVL. That is, we define a \emph{cycle}, $l$ on the graph as a sequence of nodes $(k_1, \dots k_j)$ such that $\{k_i, k_{i+1} \} \in E$ (in either direction) and $\{k_j, k_1\} \in E$.  We can find the oriented edge set of the cycle $E_l$, wherein the orientation matches the direction traversed in $l$ (these edges do not necessarily lie in $E$). We construct a cycle matrix $A$; each row represents a unique cycle in the graph, and each column an edge in $E$. The matrix entries take values of $\pm 1$ when the edge is part of the cycle and $0$ otherwise; the sign indicates whether the orientation of the edge in $E$ aligns with the direction of the cycle $(+1)$ or is opposite to it $(-1)$. The triangle graph has two possible cycles, $1\to 2 \to 3\to 1$ or $1 \to 3 \to 2 \to 1$.  We can include both for now,
\begin{equation}
A =  \bordermatrix{
& e_{12}& e_{23} & e_{13} \cr
\text{cycle}_1&1 & 1 & -1 \cr
\text{cycle}_2&-1 & -1 & 1 }
\end{equation}
These are not independent, and we must eliminate one cycle to solve for a voltage configuration, generating a reduced cycle matrix. KVL are enforced on a voltage vector $\vec{v}$ for each cycle in the graph; let $a_{l, e}$ be $1$ if the orientation matches that in the graph, $-1$ if it is opposite, and 0 if $e\notin E_l$. 
In this case, we can write,
\begin{eqnarray}
\sum_e a_{l, e} v_e = 0.
\end{eqnarray}

To solve for a voltage configuration we will have to do some extra work later to determine which cycles we include. If $\vec{v}$ is a voltage configuration then
\begin{eqnarray}
A \vec{v} = 0.
\end{eqnarray}
$A$ may be interpreted as a type of curl on the graph that gives the circulation of a vector. In this sense, the relation $A\vec{v}=0$ is analogous to saying that the voltage drops are curl-less and thus analogous to electric fields.

The current and voltage configurations belong to a set of all current and voltage configurations. In the supplementary material, we prove that these current and voltage configurations are dual. Thus, we have two valid representations of the electrical properties of the circuit. Here, we define two ways of approaching the solution of a circuit: the node method and the loop method.

In the node method, the voltage configuration is calculated by finding voltage potentials of the nodes in the circuit. One potential value is set as the ground, and the $n-1$ remaining potential values are considered as the unknowns. For the $n-1$ non-root nodes, we write down KCL and solve for the voltage potential in succession via Gaussian elimination.

In the loop method, we must have a way of choosing cycles. To do this, we introduce a spanning tree $T$. A spanning tree is the minimum number of edges to connect all nodes. The remaining cycle edges can be added to the spanning tree; when an individual cycle edge is added to the spanning tree it gives a unique cycle whose orientation we assign by that of the included edge. This is called the \emph{fundamental cycle} of the edge, also called a fundamental loop. From the reduced cycle matrix, we write down KVL. A current $\vec{i}$ may then be expressed by the \emph{fundamental loop currents} $j_e$ as a weighted sum of the fundamental cycles. Writing a \emph{reduced cycle matrix} $\tilde{A}$ in terms of the fundamental loops, we have
\begin{align}
\vec{i} = \tilde{A}^t \vec{j} .\end{align}

Given a spanning tree $T$, we can find unique paths from any arbitrary initial node, the root, and any other node $k$. We can sum the voltages at the nodes along this path; if an edge is oriented along the walk from root to node, the voltage value of that edge is added; if the walk and edge are unaligned, the voltage is subtracted. We define this walk $p_k$, and the voltage configuration can be written in terms of this fundamental walk:
\begin{align}
\vec{v} = B^t \vec{p},\end{align}
which reproduces the relation between the voltage and potential above.
These identities will be useful to derive the results below.

\subsection{Resistors and memristive devices}
A memristive device can be treated as a resistor with variable resistance; its resistance can take continuous values between a low and high resistance. The state of the resistance between two limiting values can be parameterized by a variable $x$, which is constrained by some window function such that $ 0\leq x\leq 1$. This can be thought of as the first-order term in a polynomial expansion for the resistance in terms of an adimensional parameter. Thus, the resistance varies linearly with $x$; later, we will generalize this to include higher-order terms of $x$ to describe nonlinear memresistance. In typical conditions, the resistance of a memristive device is bounded between two fixed high and low resistance states, corresponding to the 'off' and 'on' states. The resistance of these states is given by the values $\Roff\gg \Ron>0$.   
 As the resistance state depends on the history of applied voltage or bias, there is a memory stored in the resistance value, and we will refer to $x$ as the \emph{internal memory parameter}. 
 For the case of widely used metal oxide memristive devices, a simple toy model for the evolution of the resistance is the following\cite{caravelli2017complex}:
 \begin{eqnarray}
    R(x)&=&\Ron x +\Roff (1-x), \text{ the direct parametrization,}  \\
     \frac{d}{d t} x(t)&=&\frac{\Ron}{\beta} i(t)-\alpha x(t),
     \label{eq:memr1}
\end{eqnarray}
In the direct parametrization, $R(1)=\Ron$ and $R(0)=\Roff$. In equation \eqref{eq:memr1}, $i(t)$ is the current flowing in the device at time $t$, $\alpha$ is a dampening parameter with units of inverse time, and $\beta$ can be thought of as the inverse learning rate with units of time per voltage. These constants control the decay and reinforcement timescales. We define a scaling variable $\xi=\frac{\Roff-\Ron}{\Ron}$. $\alpha$, $\beta$, $\xi$ can be measured experimentally. Alternatively, a similar formulation of the resistance dynamical equation is given, 
\begin{eqnarray}
    R(x)&=&\Ron(1-x)+\Roff x, \text{ the flipped parametrization},  \\
     \frac{d}{d t} x(t)&=&\alpha x(t)-\frac{\Ron}{\beta} i(t).
     \label{eq:memr2}
\end{eqnarray}
In the flipped parametrization $R(1)=\Roff$ and $R(0)=\Ron$. Here we define the scaling factor $\chi=\frac{\Roff-\Ron}{\Roff}$ which is bounded $\chi <1$. 
These two parameterizations are inequivalent models of the memristive device dynamics. 
From the point of view of the dynamics of the single memristive device, the two parametrizations are related by the following change of variables: $\alpha\leftrightarrow -\alpha$, $\beta \leftrightarrow -\beta$ and $\xi \leftrightarrow -\chi$,  i.e., $\xi\left(\frac{-R_{\text{off}}}{R_{\text{on}}}\right)=\chi$.

\subsection{Resistance characterization}

The model given above, which has been previously studied \cite{caravelli2017complex,Caravelli2019Ent}, has a simple time dependency.
\begin{eqnarray}
    x(t)=e^{\pm \alpha (t-t_0)}x_0 \mp\frac{1}{\beta^\prime}\int^t e^{\mp\alpha (s-t)} i(s) ds
\end{eqnarray}
This is valid to a first approximation, with $(+\frac{1}{\beta^\prime})$ in the direct parameterization and $(-\frac{1}{\beta^\prime})$ for the flipped parameterization. This is the simplest description of a device that saturates between two limiting values, $\Ron$ and $\Roff$, and it can be generalized. Here, we focus on how to parametrize higher nonlinearities. This approach can also be extended to systems in which nonlinear Schottky barriers are present by modifying the effective equation for the resistance with an exponential term in front and by parametrizing the resistance with a voltage-dependent term of the form $R^\prime_n(x)=e^{\alpha V}R_n(x)$.

Consider a memristive system described by a single internal variable $x\in [0,1]$, with the property that $R(x)\in [\Ron,\Roff]$. Then, if we consider a set of resistive variables $R_{0}\leq \cdots \leq R_{k} \leq \cdots \leq R_n  $, we can write without loss of generality
\begin{equation}
    R_n(x)=\sum_{k=0}^n R_k B_{k,n}(x)
\end{equation}
where $B_{k,n}(x)$ are Bernstein polynomials with $B_{k,n}(x)=\binom{n}{k} x^k (1-x)^{n-k}$. In the limit $n\rightarrow \infty$, $R_n(x)\rightarrow  R(x)$ pointwise.
Now assume $\alpha=0$, we have that in a controlled experiment with sinusoidal inputs, and the assumption of symmetric memristive devices, we have
\begin{equation}
    \tilde R(t)=\frac{V(t)}{I(t)}= R\left(x(t)\right)=\sum_{k=0}^n R_k B_{k,n}\left(x(t)\right)
\end{equation}
and we can write
\begin{widetext}
\begin{eqnarray}
\left( \begin{array}{c}
\tilde R(t_0) \\
\tilde R(t_1) \\
\vdots \\
\tilde R(t_n)
\end{array} \right)=\left( 
\begin{array}{cccc}
B_{0,n}\left(x(t_0)\right) & B_{1,n}\left(x(t_0)\right)  & \cdots & B_{n,n}\left(x(t_0)\right)\\
B_{0,n}\left(x(t_1)\right) & B_{1,n}\left(x(t_1)\right) & \cdots & B_{n,n}\left(x(t_1)\right)\\
\vdots & \vdots & \vdots & \vdots \\
B_{0,n}\left(x(t_n)\right) & B_{1,n}\left(x(t_n)\right) & \cdots & B_{n,n}\left(x(t_n)\right)\\
\end{array}
\right)
\left( \begin{array}{c}
R_0 \\
R_1 \\
\vdots \\
R_n
\end{array} \right)
\label{eq:ideal_memristive device_dynamics}
\end{eqnarray}
\end{widetext}
from which we can infer the Bernstein parameters from the acquired data,
\begin{equation}
    \vec R=B^{-1} \tilde  R.
\end{equation}

Note that because the Bernstein polynomials are partitions of unity, one must have that if $\tilde R=1 \bar R$, then necessarily $\vec R= \bar R \vec 1$ for some vector of resistances $\bar{R}$. 

The Bernstein polynomials approximation is valid when $\tilde{R}$ is a smooth function, such that $R$ is smooth in terms of $x$. Thus, it is possible to learn a nonlinear function $R(x(t))$ given a $\vec{R}$ and real-valued resistance data $\tilde{R}(t)$. In the case of nonlinear resistance, the resistance in time-correlated real-valued samples can be expanded in terms of $\{X^{(m)}(t)\}_{m=1\ldots M}$. As the samples will be noisy in real data, and not noiseless like in the simulated ideal memristive device case, we may find deviations even from the functional form of equation \eqref{eq:ideal_memristive device_dynamics}, as shown schematically in Figure \ref{fig:bernstein}. In this case, the task will be to learn the correct function $f(X)$, where $X=\text{diag}(x)$ is an $m \times m$ diagonal matrix of memory parameters.
\begin{figure}[ht]
    \centering
    \includegraphics[scale=0.6]{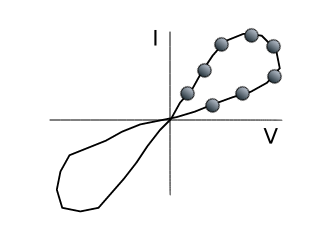}
    \caption{A schematic of an IV-hysteresis loop generated by a memristive device. The deviations from a smooth loop are due to experimental noise that is not modeled by the Bernstein polynomial approximation of the resistance. The correct function $f(X)$, without these deviations, could be determined from experimentally collected data.  }
    \label{fig:bernstein}
\end{figure}

It is interesting to note that a device with this property will still satisfy:
\begin{eqnarray}
    R(x)&=&\sum_{k=0}^n R_k B_{k,n}\left(x(t)\right)  \\
    \frac{d}{dt} x(t)&=& \alpha x(t)-\frac{1}{\beta} \frac{V(t)}{R\left( x(t)\right)},
\end{eqnarray}
with $R_0=\Ron$ and $R_n=\Roff$, and will reduce to a trivial memristive device for $n=1$. In the generic case we have a nontrivial time dependency,
\begin{eqnarray}
    x(t)=e^{\pm \alpha (t-t_0)}x_0 \mp\frac{1}{\beta}\int^t e^{\mp\alpha (s-t)} \frac{V(s)}{B(x(s))R} ds .
\end{eqnarray}
A pure memristor has $x(t)\propto q(t)$ or $x(t)\propto \phi(t)$, the charge and flux, respectively.

\subsection{Different control schemes}
In this section, we derive various control schemes for the resistance in a memristive device network. The various ways in which the memristive circuit can be controlled can be cast in terms of a generic vector $\vec{Y}$
\begin{equation}
\frac{d}{dt}\vec{x}(t)=\alpha\vec{x}(t)-\frac{1}{\beta} \left(I+\xi \Omega_A X(t)\right)^{-1} \vec{Y},
\label{eq:diffeqgen}
\end{equation}
where, as we show below, we have
\begin{eqnarray}
    \vec Y=\begin{cases}
    \Omega_A \vec S & \text{Voltage sources in series}\\
    A(A^t A)^{-1} \vec S_{ext}  & \text{Voltage sources at nodes}\\
    \Omega_B \vec j & \text{Current sources in parallel}  \\
    B^t(B B^t)^{-1} \vec j_{ext} & \text{Current sources at nodes}
    \end{cases} \; \; ,
    \label{eqn:controlschemes}
\end{eqnarray}

which is useful in various proofs concerning the control of reservoirs using memristive devices\cite{sheldonrc,BaccettiCaravelli_Arxiv_2023}. Here $\Omega_A$ and $\Omega_B$ are orthogonal projection operators, $\Omega_A=A^t(AA^t)^{-1}A$ and $\Omega_B=B^t (B B^t) ^{-1} B$. $\Omega_A$ and $\Omega_B$ are discussed further below. Further, in the supplementary material, we show that it is possible to write similar differential equations for the resistance without referring to $x$.

\subsubsection{Voltages at Nodes and in Series }

Here, we consider a network of resistors characterized by an incidence matrix $B$, cycle matrix $A$, and resistance values $r \in \mathbb{R}^m$. We also consider a set of voltage sources $s \in \mathbb{R}^m$ attached to edges in a consistent way (not violating KVL). For each edge $e\in E$, we then have
\begin{align}
v_e = r_e i_e + s_e, \quad \vec{v} = R\vec{i}+\vec{s}
\end{align}
where $R=\mathrm{diag}(r)$ is a diagonal matrix of resistances. Our goal is to solve for the current configuration $\vec{i}$.

Define $\tilde{B}$ as the \emph{reduced incidence matrix} by omitting the last row of $B$ which renders it nonsingular. We then introduce a spanning tree $T$ and reorder the edges into those belonging to the tree and those belonging to the chords,
\begin{align}
\tilde{B} = (B_T, B_C), \quad \tilde{A} = (A_T, A_C)
\end{align}
With this construction, $A_C$ is the identity. We thus have
\begin{eqnarray}
\tilde{B} \tilde{A}^t = B_T A^t_T + B_C = 0 \to A^t_T = -B_T^{-1} B_C
\end{eqnarray}
as $B_T$ is invertible by our construction. Thus,
\begin{eqnarray}
B\vec{i} = 0 = B_Ti_T + B_C i_C \to i_T = A^t_T i_C \to \vec{i} = \tilde{A}^t i_C
\label{eqn:Bi-Atic}
\end{eqnarray}
which is just the expression in terms of the fundamental loop currents we gave previously. Then, imposing KVL,
\begin{eqnarray}
\tilde{A}\vec{v} = 0 = \tilde{A} R \vec{i} + \tilde{A}\vec{s} = \tilde{A} R \tilde{A}^t i_C + \tilde{A} s .
\end{eqnarray}
This gives the full solution,
\begin{align}
\vec{i} = -\tilde{A}^t (\tilde{A} R \tilde{A}^t)^{-1}\tilde{A} \vec{s}.
\label{eq:isol}
\end{align}

Since we have seen that for voltages in series, we have $\vec s=A(A^t A)^{-1} \vec p$, it follows that we have for the voltages at nodes that $\vec{Y}=A(A^t A)^{-1} \vec S_{ext}$.

Let us now derive the memristive device network equation for voltages in series, which was previously derived\cite{caravelli2017complex}. Here, we provide a faster derivation.
We use the following form of the Woodbury matrix identity 
\begin{eqnarray}
    (P+Q)^{-1}=\sum_{k=0}^\infty (- P^{-1} Q)^k P^{-1}
    \label{eqn:expansion}
\end{eqnarray}
where assume that the matrix $P$ is invertible.

Then, let us define $P=\tilde A \tilde A^t$ and $Q=\tilde A Z \tilde A^t$. We have

\begin{eqnarray}
    \tilde A^t(\tilde A \tilde A^t+\tilde A Z \tilde A^t)^{-1}\tilde A&=&\sum_{k=0}^\infty (-1)^k \tilde A^t(( \tilde A \tilde A^t)^{-1} \tilde A Z \tilde A^t)^k (\tilde A \tilde A^t)^{-1}\tilde A\nonumber \\
    &=&\sum_{k=0}^\infty (-1)^k (\Omega_{\tilde A} Z )^k \Omega_{\tilde A}\nonumber \\
    &=&\Omega_{\tilde A} (I+Z)^{-1} \Omega_{\tilde A}=(I+\Omega_{\tilde A} Z)^{-1}\Omega_{\tilde  A} = \Omega_{\tilde{A}}\left(I+Z\Omega_{\tilde{A}}\right)^{-1}
\label{eqn:InvDerivation}
\end{eqnarray}

Here $\Omega_{\tilde{A}}$ is the non-orthogonal projection operator $\Omega_{\tilde{A}}= \tilde{A}^t\left( \tilde{A} \tilde{A}^t \right)^{-1}\tilde{A}$. In the direct parameterization we can assume $R=R_{\text{off}}-(R_{\text{off}}-R_{\text{on}}) G(X)=R_{\text{off}}(I-\chi G(X))$, where we have used from the flipped parameterization $\chi=\frac{R_{\text{off}}-R_{\text{on}}}{R_{\text{off}}}$ (which satisfies $0<\chi<1$) with $G(x)$ some function of $x$ such that $G(0)=0$ and $G(1)=1$. It follows that we can write 
\begin{subequations}
\begin{align}
\vec i &= -R_{\text{off}}^{-1}(I-\Omega_{\tilde A}(R_{\text{off}}-R_{\text{on}})G(X))^{-1} \Omega_{\tilde A} \vec s
\label{eq:isol2-1}\\
&= -R_{\text{off}}^{-1}(I-\chi\Omega_{\tilde A} G(X))^{-1} \Omega_{\tilde A} \vec s
\label{eq:isol2}
\end{align}
\end{subequations}

Then, let us use the equation for each memristive device,
\begin{eqnarray}
    \frac{dx_k}{dt}=\frac{R_{\text{off}}}{\beta^\prime} i_k-\alpha x_k ,
\end{eqnarray}
then
\begin{eqnarray}
    \frac{d\vec x}{dt}=-\frac{1}{\beta^\prime}(I-\chi\Omega_{\tilde A} G(X))^{-1} \Omega_{\tilde A} \vec s-\alpha \vec x ,
\end{eqnarray}
and if we define the voltage generators to be with the negative on the side of the memristive device, then we indicate the voltage generators with $\vec{S}$, and we have  
\begin{eqnarray}
    \frac{d\vec x}{dt}=\frac{1}{\beta^\prime}(I-\chi\Omega_{\tilde A} G(X))^{-1} \Omega_{\tilde A} \vec S-\alpha \vec x .
\end{eqnarray}

If instead in the flipped parameterization we assume $R=R_{\text{on}}+(R_{\text{off}}-R_{\text{on}}) G(X)=\Ron(I+\xi G(X))$, again with $G(0)=0$ and $G(1)=1$ and
\begin{eqnarray}
    \frac{dx_k}{dt}=-\frac{R_{\text{on}}}{\beta} i_k+\alpha x_k ,
\end{eqnarray}
we have
\begin{eqnarray}
    \frac{d\vec x}{dt}=-\frac{1}{\beta}(I+\xi\Omega_{\tilde A} G(X))^{-1} \Omega_{\tilde A} \vec S+\alpha \vec x ,
\end{eqnarray}
where $\xi=\frac{R_{\text{off}}-R_{\text{on}}}{R_{\text{on}}}>0$. Note that we have defined the activation voltages as $\frac{\beta^\prime}{R_{\text{off}}}=\frac{\beta}{R_{\text{on}}}$.

\subsubsection{Currents in Parallel and at nodes}

In analogy to the previous section, we can solve for the voltage configuration $v$, given a network of resistors and a set of current sources on edges $j \in \mathbb{R}^m$,  
where the current is added in parallel to the corresponding memristive device such that the total current through the two links is ${i} = G{v}+ {j}$, with $G$ the conductance matrix. In this case, we have
\begin{align}
B{i} = 0 = BG{v} + B{j} = BGB^t p + B{j} \to {p} = -(BGB^t)^{-1} B{j}
\end{align}
The matrix $BGB^t$ is $n-1\times n-1$ and $B$ has rank $n-1$.  We can thus invert it to get,
\begin{align}
\label{eq:vsol}
v = -B^t(BGB^t)^{-1}Bj .
\end{align}

If we drive a current $i^{ext}\in \mathbb{R}^n$ at the $n$ nodes, current conservation demands that $\sum_j i^{ext}_j = 0$.  We eliminate one node to obtain the $n-1$ dimensional $\vec{j}_\text{ext}$. This must satisfy,
\begin{align}
Bi = \vec{j}_\text{ext} = BGv = BGB^t p
\end{align}
giving
\begin{eqnarray}
    v = B^t p = B^t(BGB^t)^{-1} \vec{j}_{ext} .
\label{eqn:CurrentConservation_Injection}
\end{eqnarray}

Note the difference in sign. This is because we have usually defined currents \emph{exiting} nodes as being positive, and here, we have considered $j_{ext}$ as positive when entering the nodes.

Let us thus consider a diagonal matrix of homogeneous resistances:
\begin{eqnarray}
    R&=&\Ron \text{diag}\left(\frac{R(x_i(t))}{\Ron}\right) \nonumber \\
    &=&\Ron \left( I + \text{diag}\left(\frac{R(x_i(t))}{\Ron}-1\right) \right)
\end{eqnarray}
In terms of the Bernstein polynomials, for $n=1$ then $\frac{R(x_i(t))}{\Ron}-1$ simply reduces to $\xi X$, where $\xi=\frac{\Roff-\Ron}{\Ron}$ as previously derived (we could similarly work in the direct parametrization). We note that $\xi X$ interpolates linearly between $0$ and $\xi$.
For $n>1$, we define this function as $\xi f_n(x_i)$, which interpolates (nonlinearly in $x_i$) between $0$ and $\xi$ by the properties of the Bernstein polynomials. Thus, $f_n(x_i)$ interpolates between $0$ and $1$ and is a generalization of the $x_i$ for $n=1$. Thus, throughout the manuscript, it will be simply necessary to replace $\xi X$ with $\xi f_n(X)$ everywhere to obtain the generalized device equations, where $f_n(X)$ is a positive and bounded from above (by 1) matrix for arbitrary $n$.

\section{Memristive circuit motifs}

Here, we generalize dynamical equations for standard passive circuit elements to incorporate memristive devices and discuss the eigenvalues of such circuit motifs. Figure \ref{fig:motifs} depicts resistive, RC, and RLC devices. 

\begin{figure}[ht]
    \centering
    \includegraphics[width=.35\textwidth]{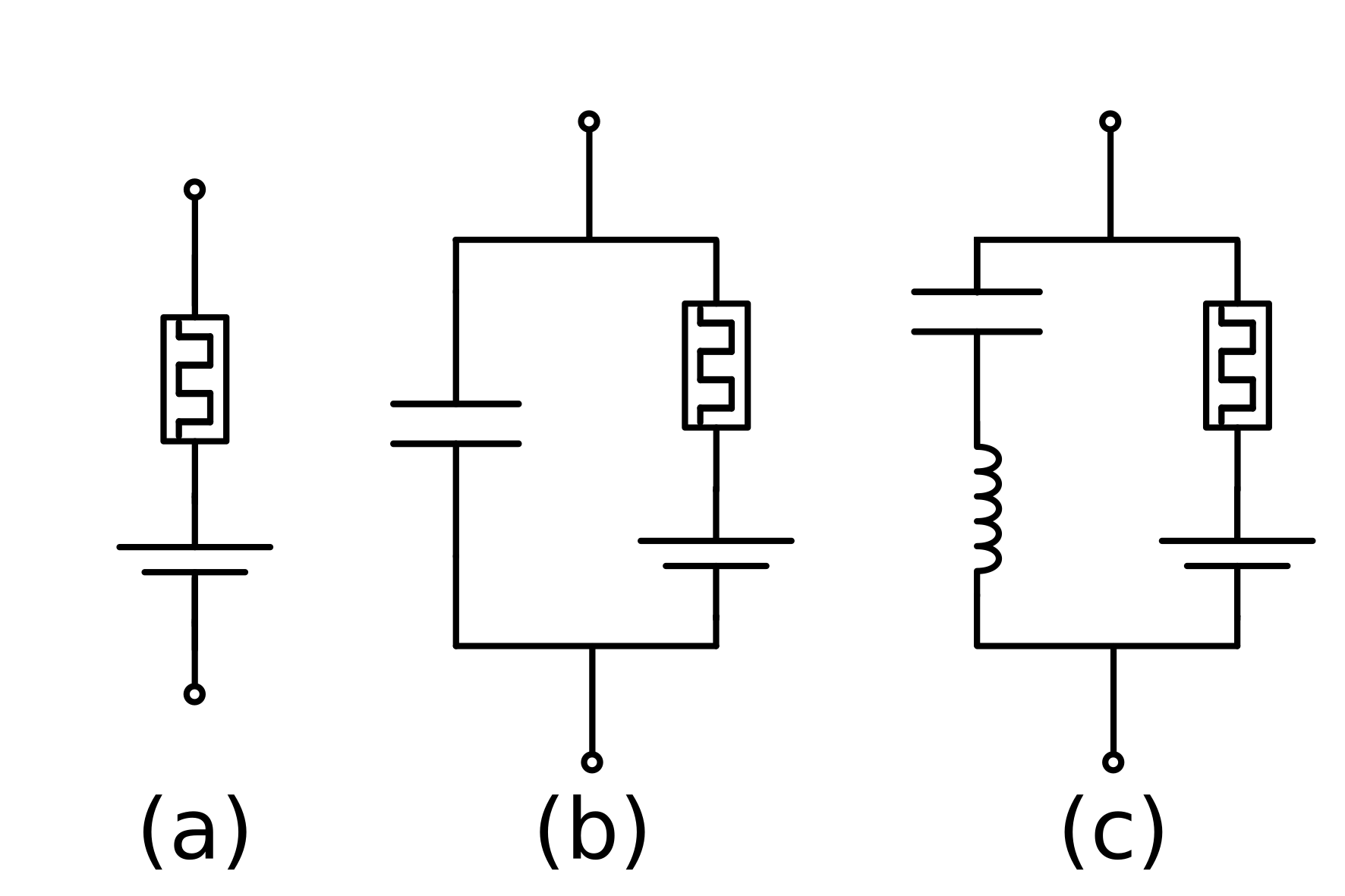}
    \caption{Schematic representations of circuit motifs. (a) Memristive device and voltage source motif. (b) The RC motif with a memristive device and voltage source in parallel with a capacitor. (c) The RLC motif with a memristive device and voltage source in parallel with an inductor and capacitor. }
    \label{fig:motifs}
\end{figure}

We examine these circuit elements as they can form fundamental motifs of more complicated circuits, going beyond purely memristive circuits. We derive equations for the dynamics of RC and RLC motifs that incorporate memristive devices. Further, we derive dynamical equations for RLC memristive device networks and describe some simple circuit structures, e.g., all circuit elements in series, which make analyzing the circuit dynamics more tractable.

Let us now try to understand the eigenvalue properties of a resistive circuit and how they depend on the device.
As shown in equation \eqref{eq:isol}, we know that the equation for the currents is of the form\cite{caravelli2017complex}
\begin{equation}
    \vec i=-A^t(A R A^t)^{-1} A \vec S_{\text{in}},
\label{eqn:iOmegaS}
\end{equation}
where $\vec{S}_\text{in}$ is the voltage applied on each edge. From equation \eqref{eq:isol2} (with $G(x)$ as function of $x$), we can rewrite this as
\begin{eqnarray}
    \vec{i}= -R_{\text{off}}^{-1}(I-\chi\Omega_{\tilde A} G(X))^{-1} \Omega_{\tilde A} \vec s .
    \label{eqn:iExpandedOmageS}
\end{eqnarray}
From equation \eqref{eqn:iOmegaS}, we can obtain the voltage across each edge
\begin{equation}
    S_\text{out}= R\vec i=-R A^t(A R A^t)^{-1} A \vec S_\text{in},
    \label{eqn:Sout=Ri}
\end{equation}
which can be written as
\begin{equation}
    S_{out}=-\Omega_R S_{in}
\end{equation}
where $\Omega_R$ is the non-orthogonal projector $\Omega_R =R A^t(A R A^t)^{-1} A $.
Thus, a resistive circuit acts simply as a filter for the voltages and is linear. We note $\Omega_R^2=\Omega_R$; thus, the eigenvalues are simply zeros and ones.

Following equation \eqref{eqn:Sout=Ri}, the memristor and source motif can be rewritten in the direct parametrization as:
\begin{eqnarray}
    S_\text{out} &=& -\left(I-\chi X(t)\right)A^t\left(A\left(I-\chi X(t)\right)A^t\right)^{-1} A S_{\text{in}}
    \nonumber\\
    &=& -\left(I-\chi X(t)\right)\left(I-\chi\Omega_A X(t)\right)^{-1}\Omega_A S_\text{in}
\end{eqnarray}
From this we can note a stability condition of $S_\text{out}$ with respect to $S_{\text{in}}$,
\begin{eqnarray}
\left(I-\chi X(t)\right)^{-1}S_\text{out}=-\left(I-\chi \Omega_A X(t)\right)^{-1}\Omega_A S_\text{in}.
\end{eqnarray}
Thus $S_\text{out}=-S_\text{in}$ when the voltage across each edge is equal to the voltage applied at each edge,
when $\left(I-\chi X(t)\right)^{-1} =\left(I-\chi \Omega_A X(t)\right)^{-1}\Omega_A $. This will occur when $\Omega_A=I$, the identity matrix, corresponding to when the cycles space $(AA^t)$ is invertible and each cycle is linearly independent.
\subsection{RC Networks}

Let us now show how this framework can be used to study standard RC circuits. Changing our network design to edges consisting of a capacitor in parallel with a resistor and voltage generator, shown in Figure \ref{fig:motifs} (b), we have at each edge,
\begin{eqnarray}
    i_e &=& i_{r,e} + i_{c,e}\\
    v_{r,e} &=& -r_e i_{r,e} + s_e = v_{c,e} = -\frac{q}{c_e} = v_e\\
    i_e &=& \frac{s_e - v_e}{r_e} - c_e\frac{dv_e}{dt},
\label{eqn:RC-current}
\end{eqnarray}
here subscripts $r$ and $c$ correspond to the resistor and capacitor edges, respectively. Imposing Kirchhoff's current law and using $\vec{v} = B^t \vec{p}$ for a potential vector $\vec{p}$
\begin{eqnarray}
    B\vec{i} &=& 0 = BR^{-1}(\vec{v} - \vec{s}) + BC\frac{d\vec{v}}{dt}\\
     BCB^t \frac{d\vec{p}}{dt} &=& -BR^{-1}(\vec{v} - \vec{s})\\
    \frac{d\vec{v}}{dt} &=& -B^t (BCB^t)^{-1}B R^{-1} (\vec{v} - \vec{s}) \nonumber \\
     &=& -\Omega_{B/C} (RC)^{-1}\vec v+\Omega_{B/C} (RC)^{-1}\vec s
    \label{eq:voltage}
\end{eqnarray}
where $\Omega_{B/C} = B^t(BCB^t)^{-1} BC$ is a non-orthogonal projector and $R,\, C$ are diagonal matrices containing the resistance and capacitance. Notably, $\Omega_{B/C}$ is asymmetric. Allowable voltage vectors satisfy $\vec{v} = B^t \vec{p}$ and are thus entirely within the row space of $B$.
The resulting dynamics are linear but do satisfy the fading memory property\cite{sheldonrc}. Considering the trajectories of two voltage configurations, subject to the same source $\vec{s}(t)$, for $\Delta \vec{v} = \vec{v}_1 - \vec{v}_2$ we have
\begin{align}
    \frac{d \Delta \vec{v}}{dt} = -\Omega_{B/C}(RC)^{-1} \Delta\vec{v}  .
\end{align}
As $\Omega_{B/C}$ has eigenvalue $1$ in the space of allowable voltages and $R_{ii}, C_{ii} > 0$, this gives exponential convergence of the trajectories.

The solution of equation \eqref{eq:voltage} is known. Let us call $\tau=rc$, we have
\begin{equation}
    \vec v(t)=e^{-\Omega_B \frac{t}{\tau}} \vec v(t_0)+\int_{t_0}^t  e^{\Omega_B \frac{(s-t_0)}{\tau}} \Omega_B (RC)^{-1}\vec{s}(s) ds .
\end{equation}
Note that the operator $e^{-\Omega_B \frac{t}{\tau}}$ can be written in a simpler form considering that $\Omega_B$ is a projector. In fact,
\begin{equation}
    e^{-\Omega_B \frac{t}{\tau}}=\sum_{k=0}^\infty (-1)^k\frac{t^k}{\tau^k k!} \Omega_B^k=I-(1-e^{-\frac{t}{\tau}})\Omega_B
\end{equation}
and thus
\begin{equation}
    e^{\Omega_B \frac{q}{\tau}}\Omega_B=e^{\frac{q}{\tau}} \Omega_B.
\end{equation}
Using the formula above, we can simply write the solution as
\begin{equation}
    \vec v(t)=\vec v(t_0)+\Omega_B\left((e^{-\frac{t}{\tau}}-1) \vec v (t_0)+\frac{1}{\tau}\int_{t_0}^t e^{\frac{s-t_0}{\tau}} \vec{s}(s) ds \right)
\end{equation}

Let's derive a memristive device version of an RC motif. We can rewrite the RC motif incorporating memristive device resistance. Plugging in the resistance from the direct parameterization, we have an equation of motion for the voltage:
\begin{eqnarray}
    \frac{d\vec{v}}{dt} &=& -\Omega_{B/C} (RC)^{-1}\vec v+\Omega_{B/C} (RC)^{-1}\vec s \nonumber\\
    &=& -\Roff^{-1}B^t(B C B^t)^{-1}B \left(I-\chi X(t) \right)^{-1} v(t)+ \Roff^{-1}B^t(BCB^t)^{-1}B\left(I-\chi X(t)\right)^{-1}\vec{s}(t) \nonumber
    \\
    &=& -\Omega_{B/C} C^{-1}\left(\vec{i}(t)+C\frac{d \vec{v}(t)}{dt} \right)\left(\vec{s}(t)-\vec{v}(t)\right)^{-1}\left( \vec{v}(t)-\vec{s}(t)\right)
\end{eqnarray}
Here, we used equation \eqref{eqn:RC-current} to substitute for the resistance. This can be rewritten,
\begin{eqnarray}
    \left(I-\Omega_{B/C} \right)\frac{d\vec{v}}{dt} &=&  \Omega_{B/C} C^{-1}\vec{i}(t) \nonumber
    \\
    &=& \Omega_{B/C} C^{-1}\beta^\prime\left(\alpha X(t)-\frac{d X}{dt}\right) .
\end{eqnarray}

Here, we have a coupled nonlinear differential equation relating voltage and the internal memory parameter. Solving for $\vec{v}(t)$ and $X(t)$ needs to be done in a self-consistent way, as $\dot{X}$ depends on both applied current and voltage sources, as well as the current and resistance in the RC motif, as determined by equation \eqref{eqn:RC-current}. Interestingly, the voltage across the capacitor will depend on the time integral of $X(t)$; thus, the voltage value has a memory of the memristive device resistance.

\subsection{RLC Networks}

The system above gives us a reservoir with real eigenvalues and some control over their distribution by choosing values of $r_e$ and $c_e$. However, the only stable reservoirs we are aware of that give extensive memories are the so-called 'resonator' systems previously studied\cite{white2004short}. In discrete time, this consists of eigenvalues spread around a circle in the complex plane with radius $\eta < 1$. In continuous time, the appropriate reservoir spectrum is given by the $z$-transform of these, giving eigenvalues with fixed real part $-\sigma$ and imaginary parts spread on the interval $[-i\pi, i\pi]$. An inductance is required to achieve imaginary eigenvalues in an electric circuit (the matrix $\Omega_{B/C}C^{-1}$ is symmetric and thus gives real eigenvalues). We begin with a sketch of the single edge case and use it as a guide for the network case.

We first consider a single edge isolated from the network, which we isolate and split into parallel edges to characterize the memristive device resistance. Our motif will consist of two edges in parallel, one with an inductor and capacitor and the other with a resistor and voltage generator. The voltages across the two edges are equal,
\begin{subequations}
\begin{align}
    v_e &= -\frac{q_e}{c_e} - l_e \frac{d i_{cl,e}}{dt} 
    \label{eqn:RLC_V-LC}
    \\
     &= -i_{r,e}r_e + s_e
     \label{eqn:RLC_Vmemristance}
     \end{align}
\end{subequations}
Here, $l_e$ and $c_e$ are the inductance and capacitance of the edge, shown in Figure \ref{fig:motifs} (c). The current through both edges is $i_e = i_{r,e} + i_{cl,e}$ and must vanish by current conservation when no current is injected across our motif. Writing these equations in terms of the voltage is inconvenient as the result will depend on $\dot{s}_e$, making comparisons with other types of reservoirs more difficult. Using the connection between the capacitive charge and current, $i_{cl,e} = \dot{q_e}$ we can write everything in terms of $q_e$,
\begin{align}
    0 = \frac{q_e}{c_e} + l_e \frac{d^2 q_e}{dt^2} + r_e \dot{q}_e + s_e.
\end{align}
This can be transformed into a first-order system,
\begin{eqnarray}
    \begin{pmatrix}
    \dot{q_e} \\
    \ddot{q}_e
    \end{pmatrix} =
    \begin{pmatrix}
    0 & 1 \\
    -\frac{1}{l_e c_e} & -\frac{r_e}{l_e}
    \end{pmatrix}
    \begin{pmatrix}
    q_e \\
    \dot{q}_e
    \end{pmatrix} +
    \begin{pmatrix}
    0 \\
    -\frac{s_e}{l_e}
    \end{pmatrix}
    \label{eqn:IsolatedRLCmotif}
\end{eqnarray}
This gives eigenvalues $\lambda_\pm = -\frac{r}{2l} \pm \sqrt{\frac{r^2}{4l^2} - \frac{1}{lc}}$.  We will obtain imaginary eigenvalues for the underdamped case, $\frac{r}{2}\sqrt{\frac{c}{l}} < 1$.

Moving to an RLC network, we can build a network consisting of RLC motifs composed as above, i.e., with memristive elements and a source parallel to an LC element; these motifs are connected by resistive edges referred to as network edges. We construct a minimum spanning tree such that the memristive edges exist on the tree and the LC components are on the cycle edges. The current in the resistive elements and the LC edges can be separated, $i_r$ and $i_{cl}$, respectively. In addition, there is a current in the edges that connect the RLC motifs; we call this the network current $i_N$, thus $\vec{i}=i_c+i_r+i_N$. Examining the cycles in the graph, due to KVL, $A\vec{v}=0$, and we will have cycles in each motif equating the voltage in the parallel meristive and LC edges, thus equations \eqref{eqn:RLC_V-LC} and \eqref{eqn:RLC_Vmemristance} remain valid in the network case. In addition, some cycles are a linear sum of memristive edges and the passive-resistive network edges. We treat the network resistance (outside the RLC motifs) as constants; the $R(X)$ matrix is now larger incorporating these constant resistive network edges that do not depend explicitly on an internal memory parameter, $x$.   
 Examining the cycles in the network, we obtain a model of the network circuit dynamics,  

\begin{eqnarray}
    A\vec{v} &=& 0 \nonumber
    \\
    &=& A R(x)i_r+A R(x) i_N -A v_c-A \vec{s} \nonumber
    \\
    &=& A R(x)A^t i_c - A R(x) i_c+A v_c-A \vec{s}
    \\
    -A^t i_c &=& -\Omega_{A/R(x)}R(x) i_c +\Omega_{A/R(x)}v_c - \Omega_{A/R(x)} \vec{s}
    \label{eqn:SimplifiedRLCcycles}
\end{eqnarray}
We have used $\vec{i}=A^t i_c$ from equation \eqref{eqn:Bi-Atic}. We have 
\begin{eqnarray}
    \Omega_{A/R(x)} v_c - \Omega_{A/R(x)}\vec{s} &=& -\Omega_{A/R(x)}R(x)i_r \nonumber
    \\
    &=& - A^t(A R(x)A^t) A R(x) i_r \nonumber
    \\
    &=& - A^t(A R(x)A^t)\left(-A R(x) i_N -A v_c+A \vec{s}\right) ,
\end{eqnarray}
from which we can rewrite equation \eqref{eqn:SimplifiedRLCcycles}
\begin{eqnarray}
    -A^t i_c &=& -\Omega_{A/R(x)}R(x) i_c +\Omega_{A/R(x)} R(x)i_N -\Omega_{A/R(x)} v_c- \Omega_{A/R(x)} \vec{s}
    \\
    0 &=& \left( \Omega_{A/R(x)}R(x)+I\right)\left(i_N+i_c\right) -\left( \Omega_{A/R(x)}R(x)-I\right)\left(-R(x)^{-1}v_c+R(x)^{-1} \vec{s}\right) \nonumber
    \\
    &=& \left( \Omega_{A/R(x)}R(x)+I\right)\left(i_N+\dot{\vec{q}}\right) -\left( \Omega_{A/R(x)}R(x)-I\right)\left(R(x)^{-1}L\ddot{\vec{q}}+(C R(x))^{-1}\vec{q}+R(x)^{-1}\vec{s}\right)
    \label{eqn:RLC_network}
\end{eqnarray}

Equation \eqref{eqn:RLC_network} is a governing equation for the RLC network. We can now examine some specific cases. 
In the case with very high resistance in the network edges connecting RLC motifs, $\lim_{R_N\rightarrow\infty} i_N\rightarrow 0$, we can simplify equation \eqref{eqn:RLC_network},
\begin{eqnarray}
    0 &=& \Omega_{A/R(x)}R(x)\left( R\dot{\vec{q}}-L\ddot{\vec{q}}-C^{-1}\vec{q}-\vec{s}\right)+ \left(\dot{\vec{q}} +R(x)^{-1} L\ddot{\vec{q}}+R(x)^{-1}C^{-1}\vec{q}+R(x)^{-1}\vec{s}\right)  .
\end{eqnarray}
The first term on the right-hand side enforces KVL within the RLC motifs with an additional constraint that the linear sum of potential across the memristive device edges, $(R(x)\dot{q}-\vec{s})$, equals zero. These are the linear sums remaining from the cycles spanning the network edges. The second term is the standard RLC dynamics in equation \eqref{eqn:IsolatedRLCmotif} for isolated RLC motifs.

When the resistance in the network edges is very low, as when the connections are of negligibly small resistance compared to the RLC elements, e.g., wires, we can approximate this as $\lim_{R_N\rightarrow 0}$. In this case, we have 
\begin{eqnarray}
    0 &=& \Omega_{A/R(x)}R(x)\left( R\dot{\vec{q}}-L\ddot{\vec{q}}-C^{-1}\vec{q}-\vec{s}\right)+ \left(i_N +\dot{\vec{q}} +R(x)^{-1} L\ddot{\vec{q}}+R(x)^{-1}C^{-1}\vec{q}+R(x)^{-1}\vec{s}\right)  .
\end{eqnarray}
The first term on the right-hand side is a voltage constraint enforcing KVL, similar to the previous case. The second term is a dynamical equation incorporating the network current $i_N$.

We examine some network topologies that simplify the circuit analysis and produce more tractable dynamical equations. If we have a network of RLC motifs in series each linked by a single network edge such that we form a ring of RLC motifs, we can again assign the LC edges to the cycle edges and the memristive edges to the tree edges. In the case of constant resistance in all the network edges, $r_N$, then we have a single cycle that spans all the memristive edges, $R(x)$ here, and the network edges,
\begin{eqnarray}
    0 &=& r_N i_N +R(x) i_r
    \\
    i_N &=& r_N^{-1}\left(L\ddot{\vec{q}} +C^{-1}q+\vec{s}\right)
\label{eqn:simpleRLCseries}
\end{eqnarray}
We can now impose $B\vec{i}=0$ using equations \eqref{eqn:simpleRLCseries},
\begin{eqnarray}
    0 &=& B(i_r+i_N+i_c) \nonumber
    \\
    &=& B(R(x)^{-1}+r_N^{-1})\left(\vec{s} +L\ddot{\vec{q}} +C^{-1}\vec{q}\right)+B\dot{\vec{q}}
    \nonumber \\
    &=& B (R^\prime)^{-1} \left(\vec{s}+L\ddot{\vec{q}}+C^{-1}\vec{q}\right) +B\dot{\vec{q}}
    \label{eqn:RLCseries-KCL}
\end{eqnarray}
Here we write $i_{cl} = \dot{\vec{q}}$ as above and have an equivalent resistance $R^\prime=(R(x)^{-1}+r_N^{-1})^{-1}$. 
To solve for $\ddot{\vec{q}}$ we use $L\ddot{\vec{q}} + C^{-1} \vec{q}= -\vec{v} = -B^t\vec{p}=R  i_r- \vec{s}$,
where $\vec{p}$ is the $n-1$ dimensional potential vector defined on the nodes, and from which
\begin{equation}
     i_r=R^{-1}[ C^{-1} \vec q+ L\ddot{\vec q}+\vec s].
\end{equation}

We can now write equation \eqref{eqn:RLCseries-KCL} as 
\begin{equation}
    0=(B(R^\prime)^{-1}B^t)^{-1} B \dot{\vec q}-\vec p+(B (R^\prime)^{-1}B^t)^{-1} B (R^\prime)^{-1}\vec s   .
\end{equation}
If we now multiply the equation above by $B^t$ on the left, we obtain
\begin{eqnarray}
    0&=&B^t(B(R^\prime)^{-1}B^t)^{-1} B \dot{\vec q}-B^t\vec p+B^t(B (R^\prime)^{-1}B^t)^{-1} B (R^\prime)^{-1}\vec s \nonumber \\
    &=&\Omega_{B/(R^\prime)^{-1}}R^\prime\dot{\vec q}+B^t\vec p+\Omega_{B/(R^\prime)^{-1}}\vec s
\end{eqnarray}
where we defined $B^t(B(R^\prime)^{-1}B^t)^{-1} B (R^\prime)^{-1}\equiv \Omega_{B/(R^\prime)^{-1}}$.

With this,
\begin{eqnarray}
    0 &=& -B(R^\prime)^{-1}(B^t \vec{p}) + B\dot{\vec{q}}  + B(R^\prime)^{-1} \vec{s} \nonumber\\
    {} &=& -\vec{p} +(B(R^\prime)^{-1}B^t)^{-1} B\dot{\vec{q}} + (B(R^\prime)^{-1}B^t)^{-1} B(R^\prime)^{-1} \vec{s} \nonumber \\
    {} &=& L \ddot{\vec{q}} + C^{-1} \vec{q} + (\Omega_{B/(R^\prime)^{-1}} R^\prime)\dot{\vec{q}} + \Omega_{B/(R^\prime)^{-1}} \vec{s}.
\end{eqnarray}

This can be written as a first-order linear system of equations,
\begin{eqnarray}
    \begin{pmatrix}
    \dot{\vec q} \\
    \ddot{\vec q}
    \end{pmatrix} =
    \begin{pmatrix}
    0 & I\\
    -(LC)^{-1} & -L^{-1}(\Omega_{B/(R^\prime)^{-1}}R^\prime)
    \end{pmatrix}
    \begin{pmatrix}
    \vec{q} \\
    \dot{\vec{q}} 
    \end{pmatrix} +
    \begin{pmatrix}
    0 \\
    -L^{-1}\Omega_{B/(R^\prime)^{-1}}\vec{s}
    \end{pmatrix}   .
    \label{eqn:RLC_matrix_1}
\end{eqnarray}

This gives eigenvalues of $\lambda_{\pm}=\frac{-\Omega_{B/(R^\prime)^{-1}}R^\prime}{2L}\pm\sqrt{\frac{\left(\Omega_{B/(R^\prime)^{-1}}R^\prime\right)^2}{4 L^2}-\frac{1}{LC}}$. We again obtain imaginary eigenvalues for the underdamped case, where $\frac{\Omega_{B/(R^\prime)^{-1}}R^\prime}{2}\sqrt{\frac{C}{L}}<1$.

We can also examine the case of a network of RLC motifs in series and without internal sources. For now, we work in the general case with nonzero sources. We incorporate the resistance from the network edges by introducing an equivalent impedance, a resistance $R_c$, on the LC edges. We assign the LC edges to the tree edges and the memristive edges to the cycle edges. The RLC motifs satisfy KVL,
\begin{eqnarray}
L\ddot{\vec{q}} + C^{-1} \vec{q} +R_c  i _c=R  i_r-\vec s .
\end{eqnarray}
Proceeding as in the previous case, we can remove the network edges and link the RLC motifs directly in series. The total current comprises only the current on the memristive and LC edges, $\vec{i}=i_r+i_c$. Now with $B(i_r+i_c)=0$, we can again solve for $\ddot{\vec{q}}$ using $L\ddot{\vec{q}} + C^{-1} \vec{q}+R_c \dot{\vec q} = -\vec{v} = -B^t\vec{p}=R  i_r-\vec s,$
where $\vec{p}$ is the $n-1$ dimensional potential vector defined on the nodes, and from which
\begin{equation}
     i_r=R^{-1}[ C^{-1} \vec q+ L\ddot{\vec q}+R_c \dot{\vec q}+\vec s].
\end{equation}

With this we have,
\begin{eqnarray}
    0 &=& -BR^{-1}(B^t \vec{p}) + B\dot{\vec{q}}  + BR^{-1} \vec{s} \nonumber\\
    {} &=& -\vec{p} +(BR^{-1}B^t)^{-1} B\dot{\vec{q}} + (BR^{-1}B^t)^{-1} BR^{-1} \vec{s} \nonumber \\
    {} &=& L \ddot{\vec{q}} + C^{-1} \vec{q} + (\Omega_{B/R^{-1}} R+R_c)\dot{\vec{q}} + \Omega_{B/R^{-1}} \vec{s}.
\end{eqnarray}

This may be cast as a first-order linear system of equations, 
\begin{align}
    \begin{pmatrix}
    \dot{\vec q} \\
    \ddot{\vec q}
    \end{pmatrix} =
    \begin{pmatrix}
    0 & I\\
    -(LC)^{-1} & -L^{-1}(\Omega_{B/R^{-1}}R+R_c)
    \end{pmatrix}
    \begin{pmatrix}
    \vec{q} \\
    \dot{\vec{q}} 
    \end{pmatrix} +
    \begin{pmatrix}
    0 \\
    -L^{-1}\Omega_{B/R^{-1}}\vec{s}  
    \end{pmatrix}  .
    \label{eq:rlc}
\end{align}

We can rewrite the RLC circuit motif by incorporating memristive devices into the resistor elements. We note the $R_c$ now depends on the resistance in the memristive device elements in a complicated way. Thus, we use $R_c(X(t))$. Here, we use the flipped parameterization and recast the system of equations with memristive device resistance $R(x)$.

\begin{align}
    \begin{pmatrix}
    \dot{\vec q} \\
    \ddot{\vec q}
    \end{pmatrix} =
    \begin{pmatrix}
    0 & I\\
    -(LC)^{-1} & -L^{-1}(\Ron \left( \Omega_A+\Omega_B\left(I+\xi X(t)\right)^{-1}\right)^{-1}\Omega_B+R_c(X(t)))
    \end{pmatrix}
    \begin{pmatrix}
    \vec{q} \\
    \dot{\vec{q}} 
    \end{pmatrix} \nonumber \\
    +\begin{pmatrix}
    0 \\
    -L^{-1}\left(\Omega_A+\Omega_B\left(I+\xi X(t)\right)^{-1}\right)^{-1}\Omega_B\left(I+\xi X(t)\right)^{-1}\vec{s}
    \end{pmatrix}
    \label{eqn:RLC-seriesLCTree}
\end{align}

If the resistance between RLC motifs is small compared to the LC impedance, then $R_c$ can be neglected; this is what we expect in most cases where good conductors, e.g., wires, connect the memristive device elements and inductors. In the case of large resistance between RLC motifs, we now examine the network conductance when the sources within the RLC motifs are zero, e.g., $s_e=0$. Then $R_c(X(t))$ reduces to $-R(x(t))$, where $R(x(t))$ is the resistance of the memristive elements. This allows us to simplify equation \eqref{eqn:RLC-seriesLCTree},
\begin{eqnarray}
    \begin{pmatrix}
    \dot{\vec q} \\
    \ddot{\vec q}
    \end{pmatrix} &=&
    \begin{pmatrix}
    0 & I\\
    -(LC)^{-1} & -L^{-1}(\Omega_{B/R^{-1}}-I)R(X)
    \end{pmatrix}
    \begin{pmatrix}
    \vec{q} \\
    \dot{\vec{q}} 
    \end{pmatrix} 
    \nonumber \\
    &=&
    \begin{pmatrix}
    0 & I\\
    -(LC)^{-1} & L^{-1}(\Omega_A+\Ron\Omega_B  R(X)^{-1})^{-1}\Omega_A R(X)
    \end{pmatrix}
    \begin{pmatrix}
    \vec{q} \\
    \dot{\vec{q}} 
    \end{pmatrix} 
    \label{eq:rlc2}
\end{eqnarray}
 In the supplementary material, we prove $(\Omega_{B/R^{-1}}-I)$ can be written in terms of the cycle projection matrix $\Omega_A$. In this case, the resistance is projected onto the cycles via $\Omega_A$. Compared with the case without memristive devices, equation \eqref{eqn:RLC_matrix_1}, there is an additional positive $L^{-1}R(x)$ term in the fourth element of the matrix. Thus, the change in the time constant $\tau=L R^{-1}$ in the cycles is a driving contribution to the current. 
  We expect the current to not propagate beyond individual RLC cycles in this network. This can be seen noting that $\ddot{\vec{q}}$ is proportionate to $(\Omega_{B/R^{-1}}-1)R(x)\dot{\vec{q}}$. The orthogonal complement to $\Omega_{B/R^{-1}}$ is related to the cycle projection matrix, $\Omega_A$, without any voltage bias within the RLC motif, the current will dissipate. We examine the eigenvalues of this system, $\lambda_\pm = \frac{-(\Omega_{B/R^{-1}}-I)R(x)}{2L} \pm \sqrt{\frac{(\Omega_{B/R^{-1}}-I)R(x)(\Omega_{B/R^{-1}}-I)R(x)}{4L^2}  -\frac{1}{LC}}$. We have imaginary eigenvalues in the under dampened case $\frac{(\Omega_{B/R^{-1}}-1)R(x)}{2}\sqrt{\frac{C}{L}} <1$, similarly $\frac{(\Omega_A+\Ron\Omega_B  R(X)^{-1})^{-1}\Omega_A R(X)}{2}\sqrt{\frac{C}{L}} <1$ .
 
 The techniques employed here can be generalized to other RLC circuit motifs with different arrangements of circuit elements.

\section{Lyapunov functions for purely memristive circuit motifs}

The Lyapunov functions provide a method to analyze the stability of memristive device networks at equilibrium while offering insight into control parameters that can be adjusted to enforce stability. The Lyapunov functions are scalar positive-definite functions, and we construct Lyapunov functions that capture the dynamics of the memristive device networks in the first-order time derivative. The conditions that ensure the time derivative of the Lyapunov functions is negative provide insight into the stable equilibrium conditions of memristive device networks. We explore the case of a dynamical equation where $R$ is linearly proportional to $x$, the case of asymmetric resistance, and when $R$ is proportional to a nonlinear function, $G_n(x)$. In addition, we study a few specific examples including the nonlinear case with $G_n(x)=\tanh{(x)}$, akin to activation functions in neural networks, and explicitly examine the role of a window function in the resistance.

\subsection{Lyapunov functions for resistance of the form $R(x)=a + x b$}
Let us discuss the Lyapunov function for the set of dynamical equations for memristive device networks, derived in earlier papers \cite{Caravelli2019Ent,caravelli2021}, but that will serve as an introduction to the generalized Lyapunov functions. We begin with
\begin{eqnarray}
\frac{d}{dt} \vec x=\alpha \vec x-\frac{1}{\beta}(I+\xi \Omega X)^{-1} \Omega \vec s, \label{eq:dyne}
\end{eqnarray}
where $X$ is a diagonal matrix of the memory parameters in the network. We will call $\Omega \vec s=\vec y$ in the following.
For the equation of motion above, we have
\begin{equation}
    (I + \xi \Omega X)\dot{\vec{x}} = \alpha \vec{x} + \alpha \xi \Omega X \vec{x} - \frac{1}{\beta} \vec{y}
\end{equation}
Consider
\begin{equation}
    L = -\frac{1}{3} \vec{x}\ ^t X \vec{x} - \frac{\xi}{4} \vec{x}\ ^t X\Omega X \vec{x}
    +\frac{1}{2\alpha\beta} \vec{x}\ ^t X\vec{y}.
\end{equation}
In this case, we have
\begin{align}
    \frac{dL}{dt} &= \dot{\vec{x}}\ ^t\left(- X\vec{x} - \xi X\Omega X\vec{x} + \frac{1}{\alpha\beta} X\vec{y} \right) \nonumber\\
    {} &= -\frac{1}{\alpha}\dot{\vec{x}}\ ^t X \big( \alpha \vec{x} + \alpha \xi \Omega X \vec{x} - \frac{1}{\beta} \vec{y}\big) \nonumber\\
    {} &= -\frac{1}{\alpha}\dot{\vec{x}}\ ^t (X + \xi X\Omega X)\dot{\vec{x}} \nonumber\\
    &= -\frac{1}{\alpha}\dot{\vec{x} }^t \sqrt{X}(I + \xi \sqrt{X}\Omega \sqrt{X})\sqrt{X}\dot{\vec{x}} \nonumber\\
    &= -\frac{1}{\alpha}||\sqrt{X}\dot {\vec{x}}||^2_{ (I + \xi \sqrt{X}\Omega \sqrt{X})} 
\end{align}
and we have that $\frac{dL}{dt} \le 0$. We can move $\dot{x}$ to the left, as all the terms in $L$ are symmetric here. Any positive semi-definite symmetric matrix, e.g., $M=I+\xi\sqrt{X}\Omega\sqrt{X}$, can be decomposed into the product of a matrix and its transpose such that $M=Q^t Q$. Thus, we can rewrite the matrix norm of the Lyapunov function,
\begin{eqnarray}
    \frac{dL}{dt} &=-\frac{1}{\alpha}\dot{\vec{x} }^t \sqrt{X}Q^tQ\sqrt{X}\dot{\vec{x}} 
    \nonumber\\
    &=-\frac{1}{\alpha}\vert\vert \sqrt{X}\dot {\vec{x}}\vert\vert^2_{Q^tQ}
    \nonumber\\
    &= -\frac{1}{\alpha}\vert\vert Q\sqrt{X}\dot {\vec{x}}\vert\vert^2 
\end{eqnarray}
As an asymptotic form, in this case, we have $x_i(\infty)=\{1,0\}$ we have 
\begin{equation}
    L =  - \frac{\xi}{4} \vec{x}\ ^t \underline{\Omega} \vec{x}
    +\vec{x}\ ^t( \frac{1}{2\alpha\beta}\vec{y} -\frac{\xi}{4}\vec\Omega -\frac{1}{3}I) .
\end{equation}
where $\underline{\Omega}$ is the matrix $\Omega$ with diagonal elements removed and $\vec{\Omega}$ is the vector of diagonal elements.

Let us now consider the Lyapunov function in the standard parameterization. We have 
\begin{eqnarray}
\frac{d}{dt} \vec x=\frac{1}{\beta}(I-\chi \Omega X)^{-1} \vec y-\alpha \vec x,\label{eq:dyne1}
\end{eqnarray}
or
\begin{eqnarray}
(I-\chi \Omega X)\frac{d}{dt} \vec x=\frac{1}{\beta} \vec y-\alpha (I-\chi \Omega X)\vec x.
\end{eqnarray}
Because of the symmetry between the two differential equations, $\alpha\beta$ remains constant and $\xi\rightarrow -\chi$. We thus attempt to write a Lyapunov function of the form: 

\begin{equation}
    L^\prime =   \big(\frac{1}{3} \vec{x}\ ^t X \vec{x} -\frac{\chi}{4} \vec{x}\ ^t X\Omega X \vec{x}
    -\frac{1}{2\alpha\beta} \vec{x}\ ^t X\vec{y}\big).
\end{equation}
Taking a time derivative, we get
\begin{align}
    \frac{dL^\prime}{dt} &=  \dot{\vec{x}}\ ^t\left( X\vec{x} - \chi X\Omega X\vec{x} - \frac{1}{\alpha\beta} X\vec{y} \right) \nonumber\\
    &= -\frac{1}{\alpha} \dot{\vec{x}}\ ^t\ X\left(\frac{1}{\beta} \vec y-\alpha (I-\chi \Omega X)\vec x \right) \nonumber\\
     &=- \frac{1}{\alpha}\dot{\vec{x}}\ ^t (X - \chi X\Omega X)\dot{\vec{x}}\nonumber\\
    &= -\frac{1}{\alpha}\dot{\vec{x} }^t \sqrt{X}(I - \chi \sqrt{X}\Omega \sqrt{X})\sqrt{X}\dot{\vec{x}} \nonumber\\
    &= -\frac{1}{\alpha}||\sqrt{X}\dot {\vec{x}}||^2_{ (I - \chi \sqrt{X}\Omega \sqrt{X})} 
\end{align}
As with the other dynamics, we see that if $I-\chi \sqrt{X} \Omega \sqrt{X}\succ 0$, the Lyapunov function always has a negative derivative. However, it is not hard to see that $\sqrt{X} \Omega \sqrt{X}\prec 1$ (if $x_i\in [0,1]$), and since $0<\chi < 1$, thus the Lyapunov property also applies in this case. Thus, the dynamics are passive and asymptotically stable for both types of circuits.

\subsection{Proof that this works only for linear functions}
The procedure we employed above for deriving the Lyapunov function only applies if the resistance $R(x)$ is linear in the internal memory parameter $x$. In order to see this, we begin again with the equations of motion,
\begin{equation}
    \frac{1}{\alpha}(I+\xi \Omega G_n(X))\frac{d\vec x}{dt}= \vec x+\xi \Omega G_n(X) \vec x-\frac{1}{\alpha \beta}  \vec y.
\end{equation}
Here $G_n(X)$ is a nonlinear function for the resistance $R$, and throughout this section $\vec{y}$ is a control vector of the memristance from equation \eqref{eqn:controlschemes}, e.g., $\Omega \vec{S}$. This is equivalent to
\begin{eqnarray}
\begin{split}
    \frac{1}{\alpha}(G_n(X)+\xi G_n(X) \Omega G_n(X))\frac{d\vec x}{dt}&= G_n(X)\vec x -\frac{1}{\alpha \beta}  G_n(X)\vec y +\xi G_n(X)\Omega G_n(X) \vec x 
    \end{split}
\end{eqnarray}
We now consider a generic Lyapunov function of the form
\begin{equation}
   L(X)=a \vec x^t F(X) \vec x + b \vec x^t  G(X) \Omega G(X)\vec x - c \vec x^t Q(X)\frac{\vec y}{\alpha \beta}
\end{equation}

with $F(X)$ a diagonal matrix
\begin{align}
\frac{dL}{dt}&=\sum_i \partial_{x_i} L(x) \partial_t x_i \nonumber \\
&=\sum_i \frac{dx_i}{dt} \Big( a\big( x_i^2 F'(x_i) + 2 F(x_i) x_i \big)  
\nonumber\\& \qquad \qquad\qquad  + b \sum_j\big( (x_i G'(x_i)+G(x_i)) \Omega_{ij} G(x_j)x_j  + x_jG(x_j) \Omega_{ji} (G(x_i)+x_i G'(x_i)) \big)  - c(Q(x_i)+x_i Q'(x_i))\frac{ y_i}{\alpha \beta}\Big)\nonumber \\
&=\sum_i \frac{dx_i}{dt} \Big( a\big(  x_i F'(x_i) + 2F(x_i) \big)x_i  + 2b \sum_j\big( (x_i G'(x_i)+G(x_i)) \Omega_{ij} G(x_j)x_j \big) - c(Q(x_i)+x_i Q'(x_i))\frac{ y_i}{\alpha \beta}\Big)
\end{align}
 Now, we will show that $\frac{dL}{dt}$ is not guaranteed to be negative by examining the case of the external field and a quadratic diagonal term. Note that for the quadratic diagonal term, we must have
\begin{equation}
    a\left(x_i F'(x_i)+2 F(x_i)\right)=G_n(x_i),
    \label{eqn:QuadraticTerm}
\end{equation}
while for the external field term
\begin{equation}
    c\left(x_i Q'(x_i)+Q(x_i)\right)=G_n(x_i).
    \label{eqn:ExternalFieldTerm}
\end{equation}
We temporarily set $(a,c)=(1,1)$. In the supplementary material, we discuss solving the quadratic and external field terms. It follows that the full solution can be written, in terms of the free parameter $a_0$  as
\begin{eqnarray}
    F(z)&=&\frac{1}{z^2} (a_0+ \int^z  q G_n(q) dq ) \\
    Q(z)&=&\frac{1}{z} (a_0+ \int^z    G_n(q) dq )
\end{eqnarray}
which is the solution to the original problem. Let us assume that $G_n(z)=a_2 z$. Then we have
\begin{equation}
    F(z)=\frac{1}{z^2}\left(a_0+\frac{a_2}{3} z^3\right)
\end{equation}
If we require that $F(0)=0$ we obtain $a_0=0$, if we require $F(1)=1$ then $a_2=3$ which leaves us with the solution $F(z)=z$. For $Q$, to enforce $Q(0)=0$ and $Q(1)=1$ we need to chose $a_0=0$ and $a_1=2$.
Thus there is a disagreement in the value of $a_2$, which can be resolved by setting $c=\frac{2}{3}$ or $(a,c)=(\frac{1}{3},\frac{1}{2})$. 

We now note that the quadratic term is problematic and that we need to ask for a function such that we obtain the symmetric term, such that we can move $\frac{dx}{dt}$ to the left. If we ask for a function $G(x_i)$ such that
\begin{equation}
   2b\left( x_i G'(x_i)+G(x_i)\right)=  G(x_i),
\end{equation}
which means 
\begin{equation}
    \partial_z \log G(z)=\frac{G(z)}{z}((2b)^{-1}-1)\rightarrow G(z)=a_0 z^{(2b)^{-1}-1} .
\end{equation}
For this reason, a symmetric quadratic term in the Lyapunov function works only in the simple case where $G(x)$ is proportional to a power of $x$, and the simplest linear memristive device ($G(x)$ linearly proportionate to $x$) occurs when $b=\frac{1}{4}$. Therefore, the Lyapunov functions found above apply only if the resistance is linear in the internal memory parameter $x$.

\subsection{Asymmetric EOMs and extension to nonlinear components}
As we will see below, there is another way of obtaining a Lyapunov function for the case of generic memristive devices by relaxing the requirement that $M$ has to be symmetric. We consider the generic equations of motion, in which $g_n(X)$ is the Bernstein polynomial approximation we introduced earlier. We have
\begin{eqnarray}
    \frac{1}{\alpha} Z(X)(I+\xi  \Omega g_n(X))\frac{d\vec x}{dt}&=& Z(X)\vec x+\xi Z(X)\Omega g_n(X) \vec x  
    -\frac{1}{\alpha \beta}  Z(X)\vec{y}
    \label{eqn:Z_lyapunov}
\end{eqnarray}
for an arbitrary nonzero diagonal matrix $Z(X)$. We consider again a generic Lyapunov function of the form
\begin{equation}
   L(X)=a \vec x^t F(X) \vec x + b \vec x^t  G(X) \Omega G(X)\vec x - c \vec x^t Q(X)\frac{\vec y}{\alpha \beta}
\end{equation}
with $F(X), Q(X)$ and $G(X)$  diagonal matrices
\begin{eqnarray}
\frac{dL}{dt}&=&\sum_i \partial_{x_i} L(x) \partial_t x_i \nonumber \\
&=&\sum_i \frac{dx_i}{dt} \Big( a\big(  x_i F'(x_i) + 2F(x_i) \big)x_i  + 2b \sum_j\big( (x_i G'(x_i)+G(x_i)) \Omega_{ij} G(x_j)x_j \big) - c(Q(x_i)+x_i Q'(x_i))\frac{ y_i}{\alpha \beta}\Big)
\end{eqnarray}

We need to solve for the linear and diagonal quadratic terms, i.e., the equations
\begin{eqnarray}
     x_i F'(x_i)+2 F(x_i)&=&\frac{1}{a}Z(x_i) \\
     x_i Q'(x_i)+ Q(x_i)&=&\frac{1}{c}Z(x_i) .
\end{eqnarray}
For the asymmetric term, we need to be careful. We need in fact to satisfy
\begin{eqnarray}
    x_i G'(x_i)+G(x_i)&=& \frac{1}{2b}Z(x_i) .
\end{eqnarray}
Since we have now the freedom of choosing $Z(x_i)$, we choose
\begin{eqnarray}
Z(x_i)&=&  2b(x_i g_n'(x_i)+g_n(x_i)) \\
G(x_i)&=&g_n(x_i),
\end{eqnarray}
so that the quadratic coupling term is immediately satisfied.
Thus, we need to solve
\begin{eqnarray}
     x_i F'(x_i)+2 F(x_i)&=&\frac{2b}{a}\left(x_i g_n'(x_i)+g_n(x_i) \right)\\
     x_i Q'(x_i)+ Q(x_i)&=&\frac{2b}{c}\left(x_i g_n'(x_i)+g_n(x_i)\right)
\end{eqnarray}
which imposes $Q(x_i)=\frac{2b}{c}g_n(x_i)$. When $(b,c)=\pm(\frac{ 1}{4},\frac{ 1}{2})$ then $Q(x_i)=g_n(x_i)$. On the other hand, we have
\begin{equation}
    F(z)=\frac{1}{z^2}\left(a_0+\frac{2b}{a} \int^z q(q g_n'(q)+g_n(q)) dq\right).
\end{equation}
Choosing $a=-\frac{1}{3}$, $b=-\frac{1}{4}$ and $c=-\frac{1}{2}$, we obtain that
\begin{equation}
    L(X)=- \frac{1}{3}\vec x^t \cdot F(X) \vec x-\frac{1}{4} \vec x^t g_n(X)\cdot \Omega g_n(X) \vec x+ \frac{1}{2}\vec x^t\cdot g_n(X) \frac{\vec y}{\alpha \beta}
\end{equation}
from equation \eqref{eqn:Z_lyapunov}, we have
\begin{eqnarray}
    \frac{d}{dt} L(X)&=&-\frac{1}{\alpha} ||\frac{d \vec x}{dt}||^2_{Z(X)+\xi Z(X) \Omega g_n(X)} \nonumber \\
    &=&-\frac{1}{\alpha} ||\frac{d \vec x}{dt}||^2_{Z(X)+\frac{\xi}{2} (Z(X) \Omega g_n(X)+ g_n(X)^t \Omega Z(X))}
    \label{eqn:Lyapunov_assymetric}
\end{eqnarray}
We arrive at the second form in equation \eqref{eqn:Lyapunov_assymetric} by noting that the inner product of real vector space is symmetric, so we can convert the matrix norm to the symmetric form. We have $Z(X)=2b\left(g_n(X)+X g_n'(X)\right)$, thus our matrix is
\begin{eqnarray}
\begin{split}
    M&=Z(X)+\xi Z(x)\Omega g_n(X) 
    \\
    &= 2b\left(g_n(X)+X g_n'(X)+\xi g_n(X)\Omega g_n(X)+\xi X g_n'(X)\Omega g_n(X)\right)
\end{split}
    \label{eqn:M-Asymmetric}
\end{eqnarray}
which is symmetric if $g_n(X)$ is linear in $X$. 

In general, we can treat $g_n(X)$ as a nonlinear function akin to the Bernstein polynomial expansion studied earlier, e.g., $g_n=x^n(1-x^n)$. Here, we examine the case in which $g_n(X)=\tanh(X)$ explicitly, as it has the right properties to give analytical results. We have $Z(X)=2b\left(X \text{sech}^2(X)+\tanh(X)\right)$. 
The function $F(z)$ is a little complicated, and is
\begin{eqnarray}
    F(z)&=&\frac{1}{z^2}\Big(a_0+\frac{2b}{a}\big(\frac{1}{2} \left(\text{Li}_2\left(-e^{-2 z}\right)-z \left(z+2 \log \left(e^{-2
   z}+1\right)\right)\right) +z^2 \tanh (z)\big) \Big)
\end{eqnarray}
where $\text{Li}_s(z)=\sum_{k=1}^\infty \frac{z^k}{k^s}$ is the polylogarithm, also called the Jonquiere's function.
We now want to impose $F(0)=0$ and $F(1)=1$.
Then we need to choose $a_0=\frac{\pi^2 b}{12 a}$, and a normalization ${\mathcal{N}}=\frac{a}{b}$
\begin{eqnarray}
\mathcal N= \left(\text{Li}_2\left(-\frac{1}{e^2}\right)-1-2 \log
   \left(1+\frac{1}{e^2}\right)\right)+\frac{\pi ^2}{12}+2\tanh (1)
   \end{eqnarray}
   The result is
   \begin{equation}
       F(z)=\frac{1}{\mathcal N}\frac{ \left(\text{Li}_2\left(-e^{-2 z}\right)-z \left(z+2 \log \left(e^{-2
   z}+1\right)\right)\right)+2 z^2 \tanh (z)+\frac{\pi ^2}{12}}{ z^2 }
   \end{equation}

Thus, the Lyapunov function depends on the eigenvalues of the matrix
\begin{equation}
    M=X \text{sech}^2(X)+\tanh(X)+\xi (X \text{sech}^2(X)+\tanh(X)) \Omega \tanh(X) .
\end{equation}
Note that we can separate $M$ into symmetric and asymmetric components, $M=M_s+M_a$, where 
\begin{eqnarray}
M_s&=&X \text{sech}^2(X)+\tanh(X)+\xi \tanh(X) \Omega \tanh(X) \\
M_a&=& \frac{\xi}{2} (X \text{sech}^2(X) \Omega \tanh(X)+\tanh(X) \Omega X \text{sech}^2(X))
\end{eqnarray}
Now, it is clear that $M_s$ is positive semi-definite. We note that on $x_i\in [0,1]$ we have 
\begin{equation}
    \rho(x_i)=\frac{x_i \text{sech}^2(x_i)}{\tanh(x_i)}\leq 1
\end{equation}
thus we write
\begin{equation}
    M_a=\frac{\xi}{2} \tanh(X)(\rho(X)\Omega+\Omega \rho(X) )\tanh(X)
\end{equation}

We see that the obtained eigenvalues are always positive, and thus, this should guarantee that the $L(X)$ constructed is a Lyapunov function.

\subsection{Case with window function}
Let us now try to further generalize the Lyapunov function for memristive equations of motion with a window function \cite{Joglekar2009,ascoli2016history}. In this case, we have
\begin{equation}
    \frac{d\vec x}{dt}=\mathcal X(X)(\alpha \vec X- \frac{1}{\beta} (I+\xi \Omega g_n(X))^{-1} \vec y)
    \label{eqn:WindowFunctionDynamical}
\end{equation}
where ${\mathcal {X}}(X)$ is the window function. We can rewrite equation \eqref{eqn:WindowFunctionDynamical} as
\begin{eqnarray}
    \frac{1}{\alpha} Z(X)\Big(I+\xi \Omega g_n(X)\Big){\mathcal X}^{-1}(X)\frac{d\vec x}{dt}= Z(X)\vec X +\xi Z(X)  \Omega g_n(X)\vec X  - Z(X)\frac{\vec y}{\alpha \beta}   .
\end{eqnarray}
We see that the problem is the same as before, but now, for the same Lyapunov function we have, 
\begin{equation}
    \frac{dL}{dt}=-\frac{1}{\alpha} \vert\vert\frac{d\vec x}{dt}\vert\vert^2_{Z(X){\mathcal X}^{-1}(X)+\xi Z(X) \Omega g_{n}(X) \mathcal X^{-1}(X)}
\end{equation}
which has to be symmetrized. Note that $Z(X)$ and $\mathcal X^{-1}(X)$ are both diagonal and thus they commute and are symmetrical. Thus we have
\begin{eqnarray}
    M&=& Z(X)\mathcal X^{-1}(X) +\frac{\xi}{2} \Big(Z(X) \Omega g_n(X)\mathcal X^{-1}(X)+ \mathcal X^{-1}(X) g_n(X)  \Omega Z(X)\Big) .
\end{eqnarray}
Note that the spectral properties of the function above are rather different. In order to see this, we can expand $Z(X)$ and write it as $Z(X)=g_n(X)+X g_n'(X)$. Then
\begin{align}
    M&= Z(X)\mathcal X^{-1}(X) +\frac{\xi}{2} \Big(g_n(X) \Omega g_n(X)\mathcal X^{-1}(X)+ \mathcal X^{-1}(X) g_n(X)  \Omega g_n(X)\Big) 
    \nonumber\\ & \qquad\qquad\qquad\qquad +\frac{\xi}{2} \Big(Xg_n'(X) \Omega g_n(X)\mathcal X^{-1}(X)+ \mathcal X^{-1}(X) g_n(X)  \Omega g_n'(X) X\Big)   .
    \label{eqn:M-Windowfnct}
\end{align}
We see that in the third term on the right-hand side, we have an operator that is not symmetric, similar to the previous case. However, the second term on the right-hand side, which was symmetric and positive before, is no longer symmetric due to the window function. This implies that we can have negative eigenvalues of order $\xi$, and in the case of window functions, the Lyapunov function we have defined is not of general applicability. However, case-by-case analysis can identify when the window function applies to the Lyapunov function. We note that ${\mathcal{X}}^-$ is diagonal, so the second row in equation \eqref{eqn:M-Windowfnct} has similar symmetry restrictions as equation \eqref{eqn:M-Asymmetric}.

\section{Invariances}

Due to the projection operators, $\Omega$, numerous circuit configurations give rise to similar properties. This is apparent by noting we can identify components of the circuit that are orthogonal to the projection operator; these orthogonal components are identified with the orthogonal projection operator $\Omega^\perp$, e.g., $\Omega_A$ and $\Omega_B$ as described in the supplementary material. This produces various local gauge symmetries, wherein the electrical properties do not vary with the components of the circuit that span the orthogonal component; this freedom is a gauge freedom. For example, for an arbitrary matrix $G(x)$ or voltage source $\vec{s}$, under the projection operator there is freedom in an orthogonal component, $\Omega_\perp \tilde{G}(x)$ and $\Omega_\perp \vec{\gamma}$, respectively. 
\begin{eqnarray}
        \Omega \left( G(x)+\Omega_\perp\tilde{G}(x)\right) &=& \Omega\left(G(x) +\left(I-\Omega\right)\tilde{G}(x)\right) \nonumber
        \\
        & =&\Omega G(x)
\\
\Omega \left(\vec{s}+\Omega_\perp \vec{\gamma}\right) &=& 
        \vec{y}
\end{eqnarray}
Equivalently, there is gauge invariance in any arbitrary vectors or diagonal matrixs $Y$ that are a projection from a ground truth ${\mathcal{Y}}$, such that $Y=\Omega{\mathcal{Y}}$:
\begin{equation*}
\begin{split}\Omega\cdot Y &=\Omega \cdot\Omega \mathcal{Y}
\\
&= \Omega \cdot\Omega^n Y .
\end{split}
\end{equation*}

The projection operator may have extensive gauge freedom in any circuit, but examining the circuit properties and memristive dynamics can restrict this freedom. For example, consider $\Omega Y= \Omega(Y\pm k\Omega^\perp \Gamma)$, with diagonal matrices $Y$ and $\Gamma$. Note $\Gamma$ here and $\vec{\gamma}$ are not necessarily related. We can devise functional forms of $Y$ that produce equivalent orthogonally projected vectors, 
\begin{equation}
\begin{split}
Y &\rightarrow \left( Y+\sum k_n \Omega^\perp \Gamma\right)
\\
&= f(\Omega^\perp,Y,\Gamma)
\end{split}
\end{equation}
where $f$ in a nonlinear function in $\Omega^\perp$, $Y$ and $\Gamma$. Consider a  nonlinear function $f(\Omega^\perp,Y,\Gamma)= \exp{\left( \Omega^\perp \Gamma\right)} Y$. As $\Omega$ is a non-orthogonal projector operator, it is necessary to track the action of the projection operators to ensure gauge freedom in an arbitrary nonlinear function. For example, consider the following two functions:
\begin{eqnarray}
\Omega f_1(\Omega^\perp,Y,\Gamma) &=&\Omega \exp{\left(\Omega^\perp{\Gamma}\right)}{Y} \nonumber
\\
&=&\Omega {Y}
\\
\Omega f_2(\Omega^\perp,Y,\Gamma)&=& \Omega {Y}\exp{\left(\Omega^\perp{\Gamma}\right)} \nonumber
\\
 &=& \Omega \sum_{n=0}^\infty \frac{{Y}}{n!}\left(\Omega^\perp{\Gamma}\right)^n \neq \Omega Y
\end{eqnarray}

Thus, the specific network properties need to be examined on a case-by-case basis to ensure invariance, as they could contribute to otherwise anomalous behaviors in the network properties. 

Finally, there is also invariance within circuit configurations; a transformation between equivalent circuits preserves the eigenspectrum of $\Omega$. The projection operators can be transformed via orthogonal matrices, $\Omega_1=O^t\Omega_2  O$; when $\Omega$ is the cycle matrix different matrices, $O$, respect different loop equivalent circuits. Since permutations are a subgroup of orthogonal transformations, they also leave the spectrum of $\Omega$ invariant (and are interpreted as edge relabeling).   

  We examine how such a transformation affects the dynamics of the memristive network. For example, we can rewrite the equations of motion with an orthogonally similar projection operator,
\begin{equation}
    \frac{d\vec x}{dt}=\frac{1}{\beta}(I+\xi
    O^t \Omega_2 O G(X))^{-1} O^t \Omega_2 O \vec S-\alpha \vec x  .
\end{equation}
We can see that while we have freedom in choosing the orthogonal matrices, $O$, restrictions on $O$ or nodal permutations may be required to preserve the network dynamics. Such restrictions will provide insight into the form of equivalent circuits. We investigate this below.

Here, we explore the effect of three invariances on the equations of motion of $x$, the current $\vec{i}$, and the Lyapunov functions. We examine invariant functions of  the types $Y+\Omega_\perp \vec{\gamma}$, $\Omega Y$ and $O^t \Omega O$. We work with a voltage source in series, $\vec{s}$, in the flipped parameterization:
\begin{eqnarray}
\vec i &= -R_{\text{on}}^{-1}(I+\xi\Omega_{\tilde A} G(X))^{-1} \Omega_{\tilde A} \vec s\label{eq:currents}
\end{eqnarray}
The results from this case are generalizable to the direct parameterization and other sources, e.g., current sources. \\
There is a more general invariance lurking in the background, which generalizes the one we just discussed. The analysis above can be used to understand the relaxational dynamics of memristors near equilibria. To see this, we consider  a network of memristive current-driven devices, whose internal memory is 
\begin{eqnarray}
   \frac{d\vec x}{dt}=-\alpha \vec x +\frac{R_\text{on}}{\beta} \vec i ,
\end{eqnarray}
which leads again to equation \eqref{eq:dyne} replacing equation \eqref{eqn:iExpandedOmageS} for $G(X)=X$. As we have seen, the transformation $\vec s\rightarrow \vec s+(I-\Omega_{\tilde A}) \vec \gamma$ for arbitrary vectors $\vec \gamma$ leaves the dynamics invariant, since $\Omega_{\tilde A}(I-\Omega_{\tilde A})=0$. This invariance does not take into account the fact that we have memristive devices. However, in memristive devices described by equation \eqref{eq:dyne} this is a subset of a larger invariance set, which is more generally a transformation of the form
\begin{eqnarray}
    (\vec x,\vec s)\rightarrow (\vec x\ ^\prime,\vec s\ ^\prime)=(\vec x+\delta \vec x,\vec s+\delta \vec s)
\end{eqnarray}
which preserves the dynamics, such that the time derivative is invariant, e.g., $\frac{d\vec x}{dt}=\frac{d\vec x\ ^\prime}{dt}$. To see this, let us equate
\begin{eqnarray}
    (I-\chi \Omega_{\tilde A} X)^{-1}\Omega_{\tilde A} \frac{\vec s}{\beta}-\alpha \vec x 
    &=&(I-\chi \Omega_{\tilde A} (X+\delta X))^{-1}\Omega_{\tilde A} \frac{(\vec s+\delta \vec s)}{\beta} -\alpha (\vec x+\delta \vec x)
\end{eqnarray}
After a few algebraic steps, we get the following generalized relationship:
\begin{eqnarray}
    (I-\Omega_{\tilde A})\delta \vec x&=&0 \\
    \delta \vec s&=&\alpha \beta \big(I-\chi (X+\delta X)\big)\delta \vec x-\chi \delta X(I-\chi \Omega_{\tilde A} X)^{-1}\Omega_{\tilde A} \vec s  +(I-\Omega_{\tilde A})\vec \gamma \label{eq:meminv}
\end{eqnarray}
for any vector $\vec \gamma$. To see the origin of this generalized invariance, consider for instance resistances $R_1,\cdots,R_n$ in series with $k$ generators $s_i$. We have
\begin{eqnarray}
    \frac{\sum_{j=1}^k s_j}{\sum_{j=1}^n R_j}= i.
\end{eqnarray}
To preserve the current flowing in the mesh, we can either change $R_j$ and $s_j$ in such a way that the ratio is preserved. For networks of memristors, equation \eqref{eq:meminv} generalizes this simple invariance. Since this is valid for arbitrary derivatives, it must be true also for the particular case of $\frac{d\vec x}{dt}=0$. As a result, this invariance maps equilibrium points as a function of a change in external control. 

\subsection{ Orthogonal vector component}
Let us examine the case where we have voltage sources with an orthogonal component applied in series, such that we can write our source as $\vec{S}+\Omega_A^\perp \vec{\gamma}$. We examine the orthogonal components' role in the time evolution of the resistance, 
\begin{equation}
\begin{split}
    \frac{d}{dt}\vec{x}(t) &=\alpha\vec{x}(t)-\frac{1}{\beta} \left(I+\xi \Omega_A X(t)\right)^{-1} \Omega_A \left(\vec{S}+\Omega_A^\perp \vec{\gamma}\right)
    \\
    &=\alpha\vec{x}(t)-\frac{1}{\beta} \left(I+\xi \Omega_A X(t)\right)^{-1} \Omega_A \vec{S}
\end{split}
\label{eqn:dotw-expandedi}
\end{equation}
The resistance in the network is independent of the orthogonal component of the source.
We now examine the Lyapunov functions of linear memristive device elements in the presence of an orthogonal contribution to the resistance. Consider the Lyapunov function,
\begin{eqnarray}
    L=-\frac{1}{3}x^t X x -\frac{\xi}{4}x^t X \Omega_A X x +\frac{1}{2\alpha\beta}x^t X \Omega_A\left(  \vec{S}+\Omega_A^\perp \vec{\gamma}\right)  ,
\end{eqnarray}
\begin{eqnarray}
    \frac{dL}{dt}=\frac{-1}{\alpha}\vert\vert \sqrt{X}x\vert\vert_{\left(I+\xi\sqrt{X}\Omega_A\sqrt{X}\right)} .
\end{eqnarray}
Again, we obtain that if $I+\xi\sqrt{X}\Omega_A\sqrt{X}\succ 0$, the Lyapunov function always has a negative derivative. As $\sqrt{X}\Omega_A\sqrt{X} >0$ in the domain $x\in\left[0,1\right]$, and $\xi>0$, the Lyapunov function has a negative derivative and the same passive dynamics and stability conditions as the system without orthogonal sources.

Similarly we can examine the case where there is a component of the resistance that is orthogonal to $\Omega_A$, $x(t)\rightarrow x^\prime(t)+\Omega^\perp_A \vec{\gamma}$, similarly $X(t)\rightarrow X^\prime(t)+\Omega^\perp_A\vec{\Gamma}$. We examine the orthogonal components' role in the equations of motion,
\begin{eqnarray}
\begin{split}  
    \dot{x}^\prime(t)+\Omega_A^\perp \dot{\vec{\gamma}} & =\alpha x^\prime(t)+\alpha\Omega_A^\perp \vec{\gamma}
   -\frac{1}{\beta}\left(I+\xi\Omega_A\left(X^\prime(t)+\Omega_A^\perp\vec{\Gamma}\right)\right)^{-1}\Omega_A\vec{s}
\end{split}
\end{eqnarray}
\begin{eqnarray}
\begin{split}  
    \dot{x}^\prime(t) &=\alpha x^\prime(t)+\Omega_A^\perp \left(\alpha\vec{\gamma} -\dot{\vec{\gamma}}\right)
     -\frac{1}{\beta}\left(I+\xi\Omega_A X^\prime(t)\right)^{-1}\Omega_A\vec{s}
\end{split}
\label{eqn:dotw_AdditiveGauge}
\end{eqnarray}
The time derivative of $x$ has an additional component projected into $\Omega_A^\perp$. Here, we note that without loss of generality, $x^\prime(t)$ has no orthogonal component as any orthogonal component can be in $\vec{\gamma}$; therefore, we can split equation \eqref{eqn:dotw_AdditiveGauge} into projected and orthogonally projected components. The orthogonal component has dissipative dynamics of the memristive device elements,
\begin{eqnarray}
    \dot{\vec{\gamma}}=\alpha\vec{\gamma} .
    \label{eqn:gammadynamcis}
\end{eqnarray}
The dissipative dynamics are short-term dynamics in the orthogonally projected space, as there could be a contribution to the dynamics of $\vec\gamma$ from the non-orthogonal component. Thus, equation \eqref{eqn:gammadynamcis} does not fix the long-term dynamics. 
Treating the dynamics of the orthogonal and non-orthogonal components as separable we see $\vec{\gamma}$ does not change the electrical current in the network.
We rewrite the current using an expansion similar to equation \eqref{eqn:iExpandedOmageS} for the flipped parameterization,
\begin{eqnarray}
    \vec{i} &=& - R_{\text{on}}^{-1}A^t\left( A A^t+ \xi A\left(x^\prime(t)+\Omega_A^\perp \vec{\gamma}\right)A^t\right)^{-1}A^t \vec{S}_{\text{in}}
    \nonumber \\
    &=&- R_{\text{on}}^{-1} \left( I+\xi \Omega_A\left(x^\prime(t)+\Omega_A^\perp \vec{\gamma}\right)\right)^{-1}\Omega_A \vec{S}_{\text{in}}
\end{eqnarray}
The orthogonal term $\Omega_A^\perp \vec{\gamma}$ does not contribute to the electrical current. Next, we examine the Lyapunov functions in the presence of an orthogonal component of the circuit. We rewrite the equations of motion, equation \eqref{eqn:dotw_AdditiveGauge},
\begin{eqnarray}
    \left(I+\xi\Omega_A X^\prime(t)\right)\dot{x}^\prime(t) - \Omega_A^\perp \left(\alpha\vec{\gamma} -\dot{\vec{\gamma}}\right)&=\alpha \left(x^\prime(t)+\xi\Omega_A X^\prime(t)x^\prime(t)
 -\frac{1}{\alpha\beta}\Omega_A \vec{s}\right).
 \label{eqn:dotw_W+Gauge}
\end{eqnarray}
We work with the transformed Lyapunov function,
\begin{align}
     L&=-\frac{1}{3}(x^\prime+\Omega_A^\perp\vec{\gamma})^t (X^\prime+\Omega_A^\perp\vec{\Gamma}) (x^\prime+\Omega_A^\perp\vec{\gamma}) 
     \nonumber \\ & \qquad\qquad\qquad\qquad   -\frac{\xi}{4}(x^\prime+\Omega_A^\perp\vec{\gamma})^t (X^\prime+\Omega_A^\perp\vec{\Gamma})  \Omega_A (X^\prime+\Omega_A^\perp\vec{\Gamma})  (x^\prime+\Omega_A^\perp\vec{\gamma})  
     +\frac{1}{2\alpha\beta}(x^\prime+\Omega_A^\perp\vec{\gamma})^t (X^\prime+\Omega_A^\perp\vec{\Gamma})  \Omega_A \vec{s}
\end{align}
We continue to work in the regime where $x^\prime(t)$ does not have a component orthogonal to $\Omega_A$ and thus is perpendicular to $\Omega_A^\perp\vec{\gamma}$. We can separate $\dot{x}^\prime$ and $\dot{\vec{\gamma}}$ in equation \eqref{eqn:dotw_W+Gauge}, and we have
\begin{eqnarray}
    \frac{dL}{dt}&=&-\dot{x^\prime}^t X^\prime\left(x^\prime+\xi\Omega_A X^\prime x^\prime-\frac{1}{\alpha\beta}\Omega_A\vec{s}\right)-\dot{\vec{\gamma}}^t\Omega_A^\perp\Gamma\Omega_A^\perp\vec{\gamma}
   \nonumber \\
    &=& -\frac{1}{\alpha}\dot{x^\prime}^t X^\prime\left(I+\xi\Omega_A X^\prime\right)\dot{x}^\prime-\frac{1}{\alpha}\dot{\vec{\gamma}}^t \Omega_A^\perp \Gamma \Omega_A^\perp \dot{\vec{\gamma}}
    \nonumber \\
    &=&  -\frac{1}{\alpha}\vert\vert \sqrt{X^\prime}\dot{x}^\prime+ \sqrt{\Gamma}\Omega_A^\perp\dot{\vec{\gamma}}\vert\vert_{I+\xi\sqrt{X^\prime}\Omega_A\sqrt{X^\prime}}
\end{eqnarray}
Where we have used the dissipative dynamics of $\gamma$ from equation \eqref{eqn:gammadynamcis}. The Lyapunov function again has a semi-definite negative time derivative. It has a similar form as the Lyapunov functions without an orthogonal source term. Thus, the memristive device networks with distinct orthogonal sources have the same passive dynamics and stability conditions.

As $x^\prime$ and $\vec{\gamma}$ are separable we can separate these into two Lyapunov functions for the orthogonal and non-orthogonal components, 
\begin{eqnarray}
 \frac{d L}{dt} &=& -\frac{1}{\alpha}\vert\vert \sqrt{X^\prime}\dot{x}^\prime\vert\vert_{I+\xi\sqrt{X^\prime}\Omega_A\sqrt{X^\prime}}
 \label{eqn:LyapunovProjected_AdditiveGauge}
 \\
 \frac{d L^\perp}{dt} &=& -\frac{1}{\alpha}\vert\vert \dot{\vec{\gamma}}\vert\vert_{\Omega_A^\perp\Gamma\Omega_A^\perp}     
 \end{eqnarray}
These two Lyapunov functions follow trivially from equation \eqref{eqn:dotw_AdditiveGauge} by applying $\Omega_A$ and $\Omega_A^\perp$. In this case, we have two equations:
\begin{eqnarray}
    \Omega_A\dot{x}^\prime(t) &=\Omega_A\alpha x^\prime(t)
     -\frac{\Omega_A}{\beta}\left(I+\xi\Omega_A X^\prime(t)\right)^{-1}\Omega_A\vec{s}
     \label{eqn:Omegadotw_additiveGauge}
\end{eqnarray}
\begin{eqnarray}
     \Omega_A^\perp \dot{\vec{\gamma}} &=& \Omega_A^\perp \alpha\vec{\gamma} .
\end{eqnarray}
Equation \eqref{eqn:Omegadotw_additiveGauge} reduces to the normal Lyapunov functions studied above, reproducing equation \eqref{eqn:LyapunovProjected_AdditiveGauge}. We can define a Lyapunov function for the orthogonal component.
\begin{eqnarray}
    L &=& -\frac{\alpha}{2}\vec{\gamma}^t\Omega_A^\perp \vec{\gamma}
    \\
     \frac{dL}{dt} &=& -\alpha\dot{\vec{\gamma}}^t\Omega_A^\perp \vec{\gamma}
     \nonumber \\
     &=& -\vert\vert \dot{\vec{\gamma}}\vert\vert_{\Omega_A^\perp} 
\end{eqnarray}
Where we have used the dissipative dynamics of $\vec{\gamma}$. Again, we see the orthogonal component has passive dynamics. 

\subsection{Projected vectors}
Here, we examine the case where we have some voltage sources applied in series that is a projection of a larger space, such that we can write our source vector as $\vec{y}=\Omega \mathcal{Y}$. This is another form of the previous case, wherein $\mathcal{Y}$ has some orthogonal component; however, now we do not know an \textit{ab initio} trivial separation into orthogonal and non-orthogonal components. We examine the equations of motion in the case of an orthogonal component in a source term,
\begin{eqnarray}
    \frac{d}{dt}\vec{x}(t) &=& \alpha\vec{x}(t)-\frac{1}{\beta} \left(I+\xi \Omega_A X(t)\right)^{-1} \Omega  \mathcal{Y} \nonumber \\
    &=& \alpha\vec{x}(t)-\frac{1}{\beta} \left(I+\xi \Omega_A X(t)\right)^{-1} \Omega \Omega \mathcal{Y} \nonumber \\
    &=& \alpha\vec{x}(t)-\frac{1}{\beta} \left(I+\xi \Omega_A X(t)\right)^{-1} \Omega \vec{y}
\end{eqnarray}
The time derivative of $x$ does not depend on the orthogonal components in $\mathcal{Y}$. Next, we examine the time derivative of $x(t)$ when there is a component of memristive devices that are orthogonal to the projection operator, $x(t)\rightarrow\Omega_A x^\prime(t)$, using a modified form of equation \eqref{eqn:dotw-expandedi}, 
\begin{eqnarray}
    \Omega_A\dot{x}^\prime(t)=\alpha\Omega_A x^\prime(t)-\frac{1}{\beta}\left(I+\xi\Omega_A\Omega_A X^\prime(t)\right)^{-1}\Omega_A\vec{s}
    \\
    \dot{x}(t)=\Omega_A\left(\alpha x^\prime(t)-\frac{1}{\beta}\left(I+\xi\Omega_A X^\prime(t)\right)^{-1} \Omega_A\vec{s}\right) .
\end{eqnarray}
The time derivative of $x(t)$is just the projection of $\dot{x}^\prime(t)$; there is no contribution from the orthogonal component. Similarly, we can see the electrical current in the network depends only on the projected component. The electrical current can be rewritten using the expansion similar to equation \eqref{eqn:iExpandedOmageS} 
 and equation \eqref{eqn:InvDerivation}, 
    \begin{eqnarray}
    \begin{split}
    \vec{i} &=- R_{\text{on}}^{-1}A^t\left( A A^t + \xi A x^\prime(t)A^t\right)^{-1}A \vec{S}_{\text{in}}
    \\
    &=- R_{\text{on}}^{-1}A^t\left( A\left( I +  \xi \Omega_A x(t)\right)A^t\right)^{-1}A \vec{S}_{\text{in}}
    \\
    &= - R_{\text{on}}^{-1} \left( I+\xi \Omega_A\Omega_A x(t)\right)^{-1}\Omega_A \vec{S}_{\text{in}}
    \\
    &= - R_{\text{on}}^{-1} \left( I+\xi \Omega_A x(t)\right)^{-1}\Omega_A \vec{S}_{\text{in}}
    \end{split}
    \end{eqnarray}

 Thus, the current in the network depends only upon the projection of $x$ onto  $\Omega_A$. This implies when the resistance is linear in $x$, then $R(X) \rightarrow \Omega_A R(X)$, and the current is invarient to $R\rightarrow \Omega_A R^\prime$. We verify this,
    \begin{eqnarray}
        \vec{i}&=& -A^t(A \Omega_A R^\prime A^t)^{-1}A\vec{S}_\text{in} 
       \nonumber \\
       &=& -R_\text{on}^{-1} A^t(A\Omega_A A^t +A\Omega_A\xi x^\prime A^t)^{-1}A\vec{S}_\text{in } \nonumber \\
       &=& -R_\text{on}^{-1} \sum \left(-A^t(A\Omega_A A^t)^{-1}A \Omega_A\xi x^\prime\right)^k A^t (A\Omega_A A^t)^{-1} A \vec{S}_{\text{in}} \nonumber\\
       &=& -R_\text{on}^{-1} \sum (1+\Omega_A \xi x )^{-1}\Omega_A \vec{S}_{in}
      \nonumber \\
       &=& -A^t(ARA^t)^{-1}A\vec{S}_\text{in}
    \end{eqnarray}

 We now examine the Lyapunov functions of linear memristive devices as a function of $x^\prime$. We have the equations of motion 
\begin{eqnarray}
    (I+\xi\Omega_A X^\prime(t))\Omega_A\dot{x}^\prime(t)=\alpha\left(\Omega_A x^\prime(t)+\alpha\xi \Omega_A X^\prime(t)\Omega_A x^\prime(t)-\frac{1}{\alpha\beta}\Omega_A\vec{s}\right)
\end{eqnarray}
We study the time derivative of a suitable Lyapunov function:
\begin{eqnarray}
    L=-\frac{1}{3} x^{\prime \, t}\Omega_A X^\prime  \Omega_A x^\prime -\frac{1}{4}\xi x^{\prime\,t}\Omega_A X^\prime \Omega_A X^\prime \Omega_A x^\prime +\frac{1}{\alpha\beta} x^{\prime\,t}\Omega_A X^\prime \Omega_A\vec{s}
\end{eqnarray}
\begin{eqnarray}
\begin{split}
    \frac{d L}{dt} &=-\dot{x}^{\prime\, t}\Omega_A X^\prime\left( \Omega_A x^\prime+\xi\Omega_A X^\prime \Omega_A x^\prime -\frac{1}{\alpha\beta}\Omega_A\vec{s}\right)
    \\
    &= \frac{-1}{\alpha}\dot{x}^{\prime \, t}\Omega_A X^\prime\left(I +\xi\Omega_A X^\prime\right)\Omega_A \dot{x}^\prime
    \\
    &= \frac{-1}{\alpha}\vert\vert \sqrt{X^\prime}\Omega_A\dot{x}^\prime\vert\vert_{(I+\xi\sqrt{X^\prime}\Omega_A \sqrt{X^\prime})}
\end{split}
\end{eqnarray}
Thus, the Lyapunov functions are projections of $\Omega_A$. As in the case with an orthogonal source, we are guaranteed that $\frac{d L}{d t}\leq 0$, but the stability conditions depend on the projection of the resistance given by $\Omega_A \dot{x}^\prime$.

\subsection{The equivalence class of spectrally invariant projectors}

It was shown in 
\cite{caravelliwein} that an orthogonal transformation on the matrix $\Omega_A$ can be interpreted as edge rewiring that leaves the cycle content of the graph invariant. Since the graph-theoretical content is contained in the projector matrix, networks that are orthogonally related have orthogonally similar projection operators. We can interrelate these projection operators by the transformation via orthogonal matrices, $O$. A simple example is shown here for a cycle matrix, $A$:

\begin{eqnarray}
    A &=& \begin{pmatrix}
        1&1&0&-1&0 \\
        0&0&1&1&-1
    \end{pmatrix}
    \\
    O&=&\begin{pmatrix}
        1&0&0&0&0 \\
        0&0&0&0&-1 \\
        0&0&1&0&0 \\
        0&0&0&1&0 \\
        0&1&0&0&0 \\
    \end{pmatrix}
    \\
    (O A^t)^t&=&\begin{pmatrix}
        1&0&0&-1&1 \\
        0&1&1&1&0
    \end{pmatrix}
\end{eqnarray}

This transformed network has an orthogonally similar projection matrix, $\Omega_A \rightarrow  O^t \Omega_A O$. We now investigate the memristive device properties of networks with orthogonally similar projection matrices, $\Omega_A$  and $O^t\Omega_A O$. We rewrite the equations of motion of $x(t)$ with this transformed projection operator for a voltage source, $\vec{S}$, in series,
\begin{eqnarray}
    \dot{x}(t)=\alpha x(t)-\frac{1}{\beta}\left( I+\xi O^t\Omega_A O x(t)\right)^{-1} O^t\Omega_A O \vec{S}  .
\end{eqnarray}
Obviously $\dot{x}(t)$ is invariant when $O^t\Omega_A O$ respects symmetries in $x$ and $\vec{Y}$ such that $O^t\Omega_A O x(t)=\Omega_A x(t)$,  $O^t\Omega_A O \vec{S}=\Omega_A\vec{S}$, this is true if $\Omega_A$ respects the orthonormal basis of $O$ and $O$ acts as a permutation operator. We gain further insight by examining the electrical current in the network, using the Woodbury matrix identity, equation \eqref{eqn:expansion}, discussed above, 
\begin{eqnarray}
    \vec{i} &=& -R_{\text{on}}^{-1}\sum_{k=0}^\infty (-1)^k \left( O^t\Omega_A O \xi x(t)\right)^k O^t \Omega_A O \vec{S}_\text{in}
    \nonumber \\
     &=& -R_{\text{on}}^{-1} O^t\sum_{k=0}^\infty (-1)^k \left(\Omega_A \xi O x(t) O^t\right)^k \Omega_A O \vec{S}_{\text{in}}
\end{eqnarray}
Thus, the transformation $\Omega_A\rightarrow O^t\Omega_A O$ is equivalent to a dual transformation:
\begin{eqnarray}
    \Omega_A\rightarrow O^t\Omega_A O \Leftrightarrow \begin{matrix}
        \vec{i}(x(t),\vec{S}_{\text{in}})\rightarrow  \vec{i}(x^\prime(t),O\vec{S}_{\text{in}}) \\
        x(t)\rightarrow x^\prime(t) \;\vert\, O x(t)O^t
    \end{matrix}
\end{eqnarray}
To examine the role of this dual transformation on the time evolution of $x(t)$, we rewrite equation \eqref{eq:memr2},
\begin{eqnarray}
        \dot{x}(t) &=& \alpha x(t)-\frac{R_{\text{on}}}{\beta} O^t A\left( A^t O R O^t A\right)^{-1}A^t O \vec{S}_\text{in}  .  
\label{eqn:wdot_gauge3}
\end{eqnarray}
Thus, the transformation $\Omega_A\rightarrow O^t\Omega_A O$ corresponds to a transformation of the orthogonal projectors with an orthogonally similar resistance matrix,
\begin{eqnarray}
\Omega_R\rightarrow O^t \Omega_{O R O^t}O ,
\end{eqnarray}
therefore $\dot{x}(t)$ is invariant to this orthogonal transformation if $\Omega_R \leftrightarrow O^t \Omega_{O R O^t}O$. 

Next, we examine the Lyapunov function of linear memristive devices under an orthogonal transformation. 
We rearrange equation \eqref{eqn:wdot_gauge3} to rewrite the equations of motion,
\begin{eqnarray}
    (I+\xi O^t\Omega_A O X)\dot{x}=\alpha\left( x+\xi O^t\Omega_A O X x -\frac{1}{\alpha\beta}O^t\Omega_A O \vec{S}_\text{in}\right)
\end{eqnarray}

We find an appropriate Lyapunov function and examine the time derivative; we have
\begin{eqnarray}
    L &=&-\frac{1}{3}x^t X x-\xi \frac{1}{4}x^t X O^t \Omega_A O X x +\frac{1}{2\alpha\beta}x^t X O^t\Omega_A O \vec{S}_{\text{in}}
\\
\frac{d L}{dt} &=&-\dot{x}^t X\left(x+\xi O^t\Omega_A O X x-\frac{1}{\alpha\beta}O^t\Omega_A O\vec{S}_{\text{in}} \right)
\nonumber \\
&=&\frac{-1}{\alpha}\vert\vert \sqrt{X}\dot{x}\vert\vert_{\left(I+\xi\sqrt{X}O^t \Omega_A O\sqrt{X}\right)}
\end{eqnarray}   
Therefore, under this transformation, the Lyapunov functions are equivalent to the original case (without orthogonal transformations) when $O\sqrt{X(t)}=\sqrt{X(t)}$, a similar restriction when $A=A O^t$; this is true when $O=I$ or when $O$ respects the symmetries of the circuits such that a permutation leaves the circuit unchanged. Therefore, the Lyapunov functions of the orthogonally transformed system are identical only if $O^t\Omega_A O=\Omega_A$. Here, the Lyapunov function has a negative derivative if $I+\xi\sqrt{X} O^t \Omega_A O \sqrt{X}\succ 0$. 

We can move the orthogonal matrix from the matrix norm to better compare the dynamics to the original case. We substitute $O X O^t\rightarrow X^\prime$, and we can rewrite the time derivative of the Lyapunov function as
\begin{eqnarray}
  \frac{dL}{dt}&=&\frac{-1}{\alpha}\vert\vert \sqrt{X^\prime} O\dot{x}\vert\vert_{\left(I+\xi\sqrt{X^\prime} \Omega_A \sqrt{X^\prime}\right)}  .
\end{eqnarray}

Thus, networks related by an orthogonal transformation, $O$, will similarly have a Lyapunov function with a negative derivative. As the inner product is invariant to orthogonal transformations, $\vert\vert \sqrt{X^\prime} O\dot{x}\vert\vert =\vert\vert \sqrt{X} \dot{x}\vert\vert $. Here, it is apparent the stability conditions are similar to the non-transformed case.

\section{Laplacian matrix and the projection operator}

We note a deep connection between the graph Laplacian matrix and the projection operators studied above. First, we prove that the graph Laplacian $L$ can be constructed from the index matrix and $G$ is the $(m \times m)$ diagonal matrix of conductivity values. 

\begin{theorem}
    The graph Laplacian can be written as $L=B G B^t$
    \label{thm:Laplacian}
\end{theorem}

\begin{proof}
    We examine the Laplacian matrix element-wise, noting that $G$ is a diagonal matrix: 
    \begin{eqnarray}
        L_{i,j} &=& \sum_k B_{i,k}B_{j,k}G_{k,k} \nonumber\\
        &=& \sum_k e^k_{i,}e^k_{j,}G_{k,k}  .
        \end{eqnarray}

    Here $e^{k}_{i,}$ indicates the incidence of edge $e^k$ on node $n_i$. As each edge is incident upon only two nodes, then if $i\neq j$, the product $e^k_{i,}e^k_{j,}$ is nonzero only if $e^k$ connects nodes $n_i$ and $n_j$. If $i=j$, then the sum is over all edges connected to node $n_i$. Therefore, we can find the elements of $L$:
        \begin{eqnarray}
        L_{i,j}&=& \begin{cases}
            -G_{k,k} \; |\;  \text{iff} \, e^k_{i,j} \in \{E\},\. i\neq j \\
            0 \; |\;   e^k_{i,j} \notin \{E\}, \. i\neq j \\
            \sum_k G_{k,k} \; |\;    \forall e^k_{i,} \in \{E\},\. i=j
        \end{cases} ,
    \end{eqnarray}
   this reproduces the graph Laplacian, completing the proof.
\end{proof}

We can immediately see that the pseudoinverse of the graph Laplacian, $(BGB^t)^+$ is related to the projector operator, $\Omega_{B/R^{-1}}$:
\begin{eqnarray}
    \Omega_{B/R^{-1}} &=& B^t(BGB^t)^+BG \nonumber \\
    &=& B^t L^+BG .
    \label{eqn:L_ProjectorProduct}
\end{eqnarray}
This retains the projection operator properties of $\Omega$,
\begin{eqnarray}
    \Omega_{B/R^{-1}}\Omega_{B/R^{-1}} &=&\Omega_{B/R^{-1}} \nonumber \\
     &=& B^tL^+BGB^tL^+BG    = B^tL^+BG  .
\end{eqnarray}

We can similarly relate $\Omega_{B/R^{-1}}$ to the Laplacian matrix,
\begin{eqnarray}
    BG\Omega_{B/R^{-1}}B^t &=& LL^+L \nonumber \\
    &=& L ,
\end{eqnarray}
and by multiplying on the left by $(BG)^+$ and on the right by $B^{t\,+}$ we can equate,
\begin{subequations}
\begin{align}
    \Omega_{B/R^{-1}} &= (BG)^+ L B^{t\, +} 
    \label{eqn:OmegaLaplacian1}\\
    &= B^tL^+BG   
    \label{eqn:OmegaLaplacian2}
\end{align}
\end{subequations}
where we have used equation \eqref{eqn:L_ProjectorProduct} to arrive at equation \eqref{eqn:OmegaLaplacian2}. Equating the right hand sides of equations \eqref{eqn:OmegaLaplacian1} and \eqref{eqn:OmegaLaplacian2}, we see that $\Omega_{B/R^{-1}}$ is an involutory matrix within the non-orthogonal projector operator subspace. In the case of the unweighted Laplacian matrix, we can follow the derivation above using $\Omega_B$ and substituting the identity matrix for $G$.

Similarly, we can relate the Laplacian matrix and $\Omega_{B/R^{-1}}$ in the case of the circuit with current injected at the nodes. Using equation \eqref{eqn:CurrentConservation_Injection} we can equate the inner products,
\begin{eqnarray}
    j_\text{ext}^tL^+ j_\text{ext}&=& \vec{i}^t\Omega_{B/R^{-1}}\vec{v},
\end{eqnarray}
 here $j_\text{ext}$ is the current injected at the nodes, $\vec{i}$ is the current configuration and $\vec{v}$ is a voltage drop across the circuit edges. By dimensional analysis we can see that this equates the system's power due to the injected current with the current and voltage configurations.

 Finally, we can follow a similar procedure to generate a matrix based on the cycle matrix $A$ related to the cycle projection operators $\Omega_A$. Using an arbitrary diagonal matrix, $G$, denoting the properties of the edges in the circuit, we define a matrix
 \begin{eqnarray}
     K &=& A G A^t \\
     K_{i,j} &=& \sum_k A_{i,k}A_{j,k}G_{k,k} .
 \end{eqnarray}
 
 Examining $K$ element-wise, when $i\neq j$, then $K_{i,j}$ is the sum of elements $G_{k,k}$ corresponding to all edges $e^k$ which are part of both cycles $c_i$ and $c_j$; the sign of $G_{k,k}$ in the sum indicates whether the orientation of cycles $c_i$ and $c_j$ are aligned or not on the edge $e^k$, $(+)$ wen the cycles are aligned, $(-)$ when they are not aligned. As we are able to separate our fundamental cycle graph $\tilde{A}$ into $A_T$ and $A_C$ sub-matrices, wherein $A_C$ is an identity matrix, then $K$ generated from $\tilde{A}$ will have no component of the cycle edges in the off-diagonal elements as the cycle edges are never in more than one cycle. The diagonal elements, $i=j$, of $K$ will consist of a sum of all elements $G_{k,k}$ which correspond to the edges $e^k$ which form the cycle $c^i$. Therefore we can define $K$ as
\begin{eqnarray}
        K_{i,j}&=& \begin{cases}
            \sum \pm G_{k,k} \; |\;  \forall \, e^k \in c_i \cap c_j,\. i\neq j \\
            0 \; |\;   c_i \cap c_j =\{\emptyset \}, \. i\neq j \\
            \sum_k G_{k,k} \; |\;    \forall e^k \in \{c_i\},\. i=j
        \end{cases}  ,
    \end{eqnarray}
which is related to a projection operator $\Omega_{A/G}$ as 
\begin{eqnarray}
    \Omega_{A/G} = GA^tK^+A .
\end{eqnarray}
In the supplementary material, we discuss how this matrix can be used to calculate the power in the cycles of a circuit.

\section{ Effective resistance and memristive device correlations}

In experimental studies of memristor networks it is often necessary to characterize the networks resistance without accessing all of the individual memristors. Instead one is often restricted to two-point or four-point probe measurements to characterize the network, in which case the effective resistance between contacts is being measured. Thus it is important to understand how the current and resistance in distant but accessible edges are related. This understanding could enable new methods to control the resistance between input and output edges. 

Here, we use the effective resistance to study correlations between memristive device elements. The effective resistance can be calculated between individual edges, e.g., memristive devices, in the network. In the supplementary material, we develop a correlation analysis in the case where charge moves through a circuit via a random walk mechanism. Here, we generalize the correlation analysis to circuits under applied bias. By finding effective resistances between nodes in a network, we can examine the correlation of electrical current in two distant, non-adjacent, edges by constructing an effective Wheatstone bridge circuit.

\subsubsection{Two-point effective resistance}

We begin by finding the effective resistance between nodes. The effective resistance, $R^\text{eff}$, between the $n_a$ and $n_b$ can be found by examining the $(a,b)$ walks and removing the closed cycle walks which return to $n_a$, $(a,a)$, and similarly examining the reverse walk from $(b,a)$: 
\begin{eqnarray}
R^{\text{eff}}_{(a,b)}(a,b) &=& -L^+(a,b)-L^+(b,a)+L^+(a,a)+L^+(b,b) .
\label{eqn:2point_effectiveresistance_main}
\end{eqnarray}
The effective resistance can be considered the resistance connecting two distant nodes via a hypothetical edge, it is the resistance such that if we apply a single volt between $n_a$ and $n_b$ the measured current $i_{a,b}$ at our generator would be $\frac{1}{R^{\text{eff}}_{(a,b)}}$. This is the definition of the two-point effective resistance, which is motivated in the supplementary material by  a Markov process for particle movement in a circuit. We rederive the two-point effective resistance below. We can similarly find the effective resistance by constructing an incidence matrix for the hypothetical edge we are interested in. We construct a vector, $\vec{i}_{a,b}$, that represents sending a single particle from $n_a$ to $n_b$ and is charge balanced otherwise. We define
\begin{eqnarray}
    B \vec{i}_{a,b}={\mathfrak{b}}_{(a,b)} ,
\end{eqnarray} 
here ${\mathfrak{b}}_{(a,b)}$ is an $n$-vector that is $1$ in $a$, $-1$ in $b$, and zero otherwise. We can write
\begin{eqnarray}
    R^{\text{eff}}_{(a,b)}= {\mathfrak{b}}^t_{(a,b)} L^+ {\mathfrak{b}}_{(a,b)} 
\end{eqnarray}

From which we can see a transfer current projector operator reproduces the node projection operator, \begin{eqnarray}
    \Omega_{B/G} G^{-1}&=& B^t L^+ B
    \\
    R^\text{eff}_{a,b}&=& \vec{i}_{a,b}^t\Omega_{B/G} G^{-1} \vec{i}_{a,b} .
\end{eqnarray}
Thus the effective resistance is related to projection operators by the single particle current vectors.  As $\Omega_{B/G}G^{-1}$ is symmetric, the two point effective resistance is the positive semi-definite inner product of the single particle current vectors.

\subsubsection{Effective resistive circuit}
If $e_{a,b}$ and $e_{c,d}$ are edges in our network ($e^0$ and $e^1$ respectively), with corresponding currents $i_0$ and $i_1$, then we can define a correlation between $i_0$ and $i_1$. We define an effective circuit with effective resistors between the four nodes that define the edges $e^0$ and $e^1$, $R_{a,c}$, $R_{a,d}$, $R_{b,c}$, and $R_{b,d}$; here we drop the $_\text{eff}$ superscript. This allows us to construct a simple graph consisting of the edges $e_{a,b}$ and $e_{c,d}$ ($e^0$ and $e^1$) with their effective resistance $R^\text{eff}_{0}$ and $R^\text{eff}_{1}$ (distinct from the true resistance values $R_0$ and $R_1$) and the effective edges $e_{a,c}$, $e_{a,d}$, $e_{b,c}$, and $e_{b,d}$ with their calculated effective resistance. A schematic of this graph is shown in Figure \ref{fig:effectiveResistance}.

\begin{figure}[ht]
    \centering
    \includegraphics[width=.6\textwidth]{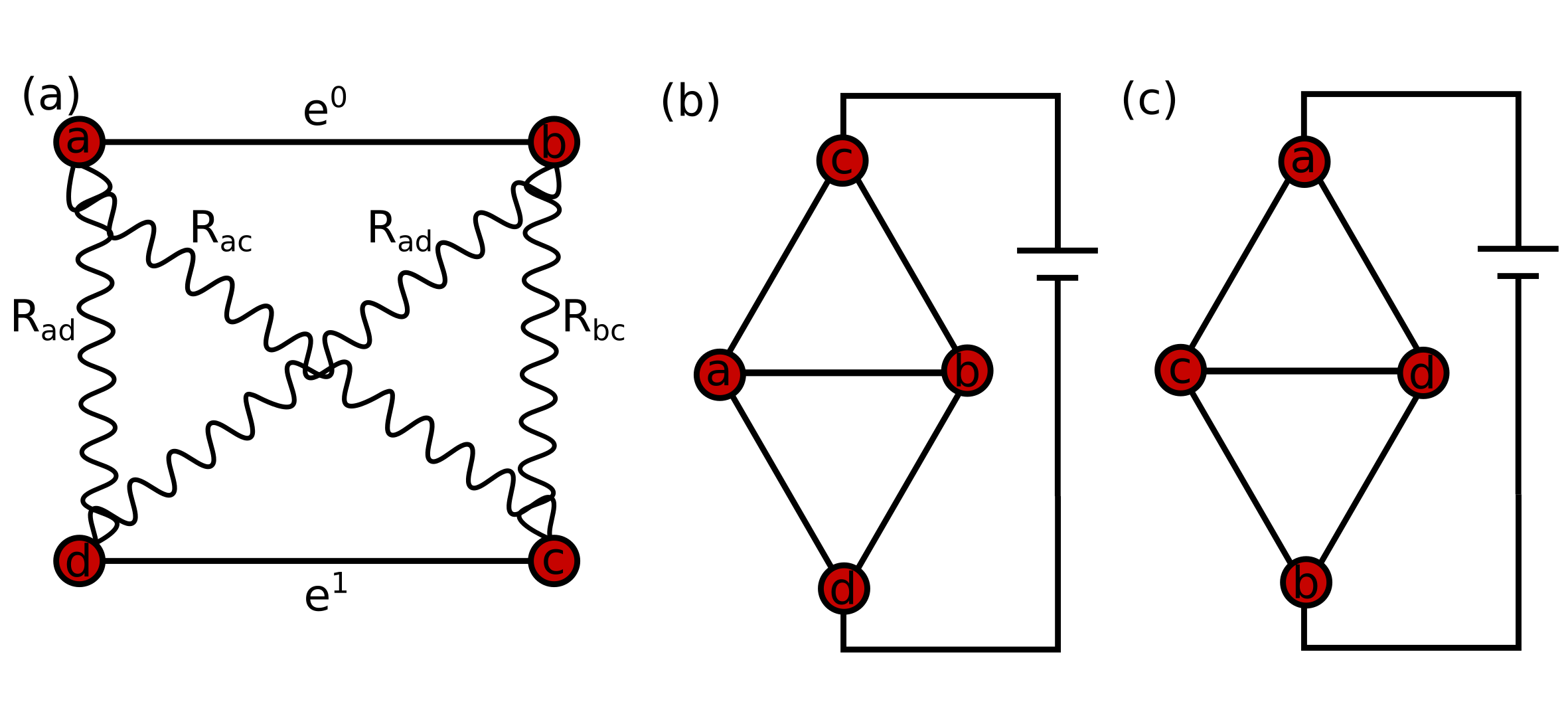}
    \caption{(a) The effective circuit linking $n_a$ and $n_b$ to $n_c$ and $n_d$ with six effective edges with effective resistances. (b) and (c) two representations of the unbalanced Wheatstone bridge circuit with sources $\Delta V_1$ and $\Delta V_0$.}
    \label{fig:effectiveResistance}
\end{figure}
To find the effective resistance of each edge in our effective circuit, we need to calculate the resistance between each effective edge that preserves the power dissipation in the network. 
This is a generalization of the Star-Delta transformation of resistors. We can find such a transformation to an effective circuit by considering a partition function $Z$ defined by a Gaussian model of the total power in the circuit. Here $\phi$ is the potential on the nodes of our full network, and the power in the system is defined as $P=\sum_{i,j} (\phi_i-\phi_j)^2/R_{ij}$. We define a partition function for the power of our circuit and integrate out the electric potential from all but the four nodes that remain in our effective circuit, $\{n_a.n_b,n_c,n_d\}$. 

We define,
\begin{eqnarray}
    Z &=&\int d\phi_0 \cdots d\phi_{n-4} \exp{\left(\frac{-\beta}{2}\sum_{i,j}\frac{(\phi_i-\phi_j)^2}{R_{i,j}}\right)} 
    \nonumber \\
    &=& \int d\phi_B \exp{\left( \frac{-\beta}{2}\vec{\phi}^t L \vec{\phi}\right)} \nonumber
    \\ &=& \int  d\phi_B \exp{\left(\frac{-\beta}{2} \phi_B^t L_{BB} \phi_B+\phi_S^t L_{SS} \phi_S +\phi_S^t L_{SB}\phi_B+\phi_B^t L_{BS} \phi_S\right)}
\end{eqnarray}
where $\phi_B$ is $\vec{\phi}_{0\cdots n-4}$, the vector of electric potentials excluding the four nodes in the effective circuit, e.g., the bulk in the network, and $\phi_S= \vec{\phi}_{n-3\cdots n}$, e.g., the surface of the network. We have written the power in terms of the Laplacian matrix for the circuit, $L$. Here we have rearranged our Laplacian and $\vec{\phi}$ such that $\{n_a,n_b,n_c,n_d\}$ are the last nodes in our vector $\vec{\phi}$. The matrices $L_{BB}$, $L_{BS}$, $L_{SB}=L_{BS}^t$, and $L_{SS}$ are the block elements of the Laplacian matrix:
\begin{eqnarray}
    L=\begin{pmatrix}
        L_{BB} & L_{BS}
        \\ L_{BS}^t &L_{SS}
    \end{pmatrix}.
\end{eqnarray}
$L_{BB}$ is a $(n-4)\times(n-4)$ symmetric submatrix, $L_{BS}$ is a $4\times(n-4)$ submatrix, and $L_{SS}$ is $4\times4$ submatrix.
Performing the integral,
\begin{eqnarray}
    Z=\sqrt{\frac{(2\pi)^{n-4}}{\beta^{n-4}\det(L_{BB})}}\exp{\left( \frac{\beta}{2}\phi_S^t\left(L_{BS}^tL_{BB}^{-1}L_{BS}-L_{SS}\right)\phi_S\right)} ,
    \label{eqn:SimplifiedCorrelation_4}
\end{eqnarray}
here $L_{BB}$ must be a real positive-definite matrix, as $L_{BB}$ is a submatrix of $L$ it is not necesarily a positive semi-definite matrix. We can rewrite our Gaussian integral in terms of the bottom $4\times 4$ block, e.g., the $S\times S$ block, of the inverse Laplacian matrix\footnote{Note that this is the result is the same as the bottom corner of a block-wise matrix inversion (Schur inverse),
\begin{eqnarray}
    \begin{pmatrix} A & B \\ C & D \end{pmatrix}^{-1}= 
    \begin{pmatrix}
        A^{-1}+A^{-1}B\left(D-CA^{-1}B\right)^{-1}CA^{-1} & -A^{-1}B\left(D-CA^{-1}B\right)^{-1}
        \\
        -\left(D-CA^{-1}B\right)^{-1} C A^{-1} & \left(D-CA^{-1}B\right)^{-1}
    \end{pmatrix}
\end{eqnarray}}:
 \begin{eqnarray}
     Z\propto \exp{\left( \frac{-\beta}{2}\phi_S^t (L^+_{4\times4})^+\phi_S\right)}.
 \end{eqnarray}
Here $(L^+_{4\times4})^+$ is the effective Laplacian for our effective circuit, $L^{\text{eff}}$; it is symmetric, but in general, the rows and columns do not sum to $0$ and is not necessarily a singular matrix. We find the exact effective resistance matrix,  
\begin{eqnarray}
    L^{\text{eff}} &=& B G^\text{eff}B^t \\
    L^+_{4x4} &=& \left(B G^\text{eff}B^t\right)^+
    \\
    G^{\text{eff}+} &=& B^t L^+_{4\times 4} B \nonumber\\
    &=& \tilde{R}^{\text{eff}} ,
\end{eqnarray}
where $B$ is the incidence matrix for our effective circuit, which we define 
    \begin{eqnarray}
B = \bordermatrix{
  & e^{0} & e_{ad} & e_{ac} & e_{bd} & e_{bc} & e^1 \cr
  n_a &1  &  -1 &  1 & 0 & 0 & 0 \cr
n_b &-1 &  0 &  0 & -1 & 1 & 0 \cr
n_c &0  & 0 & -1 & 0 & -1  & 1\cr
n_d & 0 & 1 & 0 & 1 & 0 & -1 \cr
}   .
\end{eqnarray} 
We can write $\tilde{R}^\text{eff}$ in terms of the elements of the $L^+_{4\times 4}$ matrix:
\begin{align}
   & \tilde{R}^\text{eff}= \nonumber \\ & \colvec{
L^+_{aa} - L^+_{ab} - L^+_{ba} + L^+_{bb} & - L^+_{aa} + L^+_{ad} + L^+_{ba} - L^+_{bd} & L^+_{aa} - L^+_{ac} - L^+_{ba} + L^+_{bc} & - L^+_{ab} + L^+_{ad} + L^+_{bb} - L^+_{bd} & L^+_{ab} - L^+_{ac} - L^+_{bb} + L^+_{bc} & L^+_{ac} - L^+_{ad} - L^+_{bc} + L^+_{bd}
\\- L^+_{aa} + L^+_{ab} + L^+_{da} - L^+_{db} & L^+_{aa} - L^+_{ad} - L^+_{da} + L^+_{dd} & - L^+_{aa} + L^+_{ac} + L^+_{da} - L^+_{dc} & L^+_{ab} - L^+_{ad} - L^+_{db} + L^+_{dd} & - L^+_{ab} + L^+_{ac} + L^+_{db} - L^+_{dc} & - L^+_{ac} + L^+_{ad} + L^+_{dc} - L^+_{dd}
\\L^+_{aa} - L^+_{ab} - L^+_{ca} + L^+_{cb} & - L^+_{aa} + L^+_{ad} + L^+_{ca} - L^+_{cd} & L^+_{aa} - L^+_{ac} - L^+_{ca} + L^+_{cc} & - L^+_{ab} + L^+_{ad} + L^+_{cb} - L^+_{cd} & L^+_{ab} - L^+_{ac} - L^+_{cb} + L^+_{cc} & L^+_{ac} - L^+_{ad} - L^+_{cc} + L^+_{cd}
\\- L^+_{ba} + L^+_{bb} + L^+_{da} - L^+_{db} & L^+_{ba} - L^+_{bd} - L^+_{da} + L^+_{dd} & - L^+_{ba} + L^+_{bc} + L^+_{da} - L^+_{dc} & L^+_{bb} - L^+_{bd} - L^+_{db} + L^+_{dd} & - L^+_{bb} + L^+_{bc} + L^+_{db} - L^+_{dc} & - L^+_{bc} + L^+_{bd} + L^+_{dc} - L^+_{dd}
\\L^+_{ba} - L^+_{bb} - L^+_{ca} + L^+_{cb} & - L^+_{ba} + L^+_{bd} + L^+_{ca} - L^+_{cd} & L^+_{ba} - L^+_{bc} - L^+_{ca} + L^+_{cc} & - L^+_{bb} + L^+_{bd} + L^+_{cb} - L^+_{cd} & L^+_{bb} - L^+_{bc} - L^+_{cb} + L^+_{cc} & L^+_{bc} - L^+_{bd} - L^+_{cc} + L^+_{cd}
\\L^+_{ca} - L^+_{cb} - L^+_{da} + L^+_{db} & - L^+_{ca} + L^+_{cd} + L^+_{da} - L^+_{dd} & L^+_{ca} - L^+_{cc} - L^+_{da} + L^+_{dc} & - L^+_{cb} + L^+_{cd} + L^+_{db} - L^+_{dd} & L^+_{cb} - L^+_{cc} - L^+_{db} + L^+_{dc} & L^+_{cc} - L^+_{cd} - L^+_{dc} + L^+_{dd}}
\end{align}
The power dissipation of this system is now $P=\Delta \phi_S^t \tilde{R}^{\text{eff}+}\Delta \phi_S$, where $\Delta \phi_S$ is the difference electric potential, $\Delta \phi_S=B^t\phi_S$. Importantly, the diagonal elements are the effective resistance between pairs of nodes, equation \eqref{eqn:2point_effectiveresistance_main}, the off-diagonal terms are cross talk which defines a power loss between current in distinct edges, e.g., $P_{i,j}=\vec{i}_i \tilde{R}^{\text{eff}}_{i,j}\vec{i}_j$. Due to the presence of these off-diagonal terms in $\tilde{R}^\text{eff}$, it is more convenient to determine the effective conductance, and thus resistance, defined by the off-diagonal terms of $L^\text{eff}$, e.g., $R_{ab}=\left(L^\text{eff}_{a,b}\right)^{-1}$. In the supplementary material, we discuss a mean field approximation of the power dissipation in our effective circuit using the effective resistance $\tilde{R}^\text{eff}$.   

The above methods can be extended to reproduce the effective resistance between two nodes. We can rewrite equation \eqref{eqn:SimplifiedCorrelation_4} as 
\begin{eqnarray}
    Z &=&\sqrt{\frac{(2\pi)^{n-4}}{\beta^{n-4}\det(L_{BB})}}\int d\phi_a d\phi_b \exp{\left( -\frac{\beta}{2}[\phi_a,\phi_b,\phi_c,\phi_d] (L^+_{4\times4})^+ [\phi_a,\phi_b,\phi_c,\phi_d]^t\right)} 
    \nonumber \\
    &\propto& \exp{\left( -\frac{\beta}{2}[\phi_c,\phi_d] (((L^+_{4\times4})^+)^+_{2\times2})^+ [\phi_c,\phi_d]^t\right)}   ,
\end{eqnarray}
 the matrix in the exponent can be simplified $(((L^+_{4\times4})^+)^+_{2\times2})^+=(L^+_{2\times2})^+$. Using a two-point incident matrix, $B$, the two-point effective resistance is written in term of the $L^+_{2\times2}$ as:
\begin{eqnarray}
    \tilde{R}^\text{eff}_{2\times2}=L^+_{cc}+L^+_{dd}-L^+_{cd}-L^+_{dc}   ,
\end{eqnarray}
reproducing the two-point effective resistance of equation \eqref{eqn:2point_effectiveresistance_main}.
\subsubsection{Current correlations}
From these exact effective resistance values, we can solve for the current correlations in our effective circuit. The relation between the current in the edges can be determined via standard circuit graph techniques described above and by examining the sum of spanning trees that link $n_a$ and $n_b$ while spanning $e^1$ (see for instance \cite{Biggs_AlgGraphTheory_1974}). Here, we provide a different method to find correlations between currents without knowing the full network connectivity. We note that our effective circuit, which contains all possible paths between edges $e^0$ and $e^1$, can be redrawn as an unbalanced Wheatstone bridge circuit, with voltage sources of $\Delta V_0 = V_a - V_b$ and $\Delta V_1 = V_c-V_d$, respectively, as shown in Figure \ref{fig:effectiveResistance} (b) and (c).

We define a correlation for a unit of current in $e^0$, 
\begin{eqnarray}
\langle i_1, i_0\rangle=\frac{R_{a,d} R_{c,b}- R_{a,c} R_{b,d}}{(R_{a,d}+R_{b,d})(R_{a,c}+R_{c,b})R_1+
R_{a,d} R_{a,c} R_{b,c}+R_{a,c} R_{d,b} R_{b,c} + R_{a,d} R_{a,c} R_{d,b} + R_{a,d} R_{d,b} R_{b,c} }R_{0} \frac{R^\text{eff}_1}{R_1}i_0^2  .
\label{eqn:correlation}
\end{eqnarray}
Here $R_{a,b}$ are effective resistance values obtained from $L^\text{eff}_{a,b}$ and $L^\text{eff}_{b,a}$, e.g., by averaging.

We arrive at this correlation by noting that we can construct multiple subgraphs using the six effective resistances. We construct two unbalanced Wheatstone bridge circuits, and from these circuits we find a mapping that relates the voltages of the circuit nodes. If we know $V_c$ and $V_d$ we can find $V_a$ and $V_b$. i.e., $(V_c,V_d)\rightarrow (V_a,V_b)$. We find a mapping from $\Delta V_0$ to $\Delta V_1$, the forward map, and the inverse map from $\Delta V_1$ to $\Delta V_0$. We do this by removing the edge outside the Wheatstone bridge and solving for the voltage configurations, this is detailed in the supplementary material.
In the supplementary material, we explicitly show the two linear transformations for each unbalanced Wheatstone bridge configuration. The two linear mappings are,
\begin{eqnarray}
   \begin{pmatrix}
        V_c \\ V_d
    \end{pmatrix}
    &=& 
    Q
    \begin{pmatrix}
    V_a \\ V_b     
    \end{pmatrix}
     \\
      \begin{pmatrix}
        V_a \\ V_b
    \end{pmatrix}
    &=& 
    M 
    \begin{pmatrix}
    V_c \\ V_d     
    \end{pmatrix}  ,
\end{eqnarray}
we provide exact expressions for the matrix elements using the relationship $M=Q^{-1}$ in the supplementary material.

From these two linear mappings, we find a scaling relation between $i_0$ and $i_1$. There exist linear mappings for the currents:
\begin{eqnarray}
    f\left(Q\frac{R_0}{R_1} \right)&:& i_0 \rightarrow i_1
    \\
    f\left( Q^{-1} \frac{R_1}{R_0}\right) &:& i_1 \rightarrow i_0
\end{eqnarray}
The determinant of $Q$ and $ Q^{-1}$ give us a scaling relation for the current transformation, this is detailed in the supplementary material. We have
\begin{align}
\det{ Q\frac{R_{0}}{R_{1}}} &= \frac{R_{a,d} R_{c,b}- R_{a,c} R_{b,d}}{(R_{a,d}+R_{b,d})(R_{a,c}+R_{c,b})R^\text{eff}_1+
R_{a,d} R_{a,c} R_{b,c}+R_{a,c} R_{d,b} R_{b,c} + R_{a,d} R_{a,c} R_{d,b} + R_{a,d} R_{d,b} R_{b,c} }R_{0}\frac{R^\text{eff}_1}{R_1}
\label{eqn:TransformationDeterminant1}
\\
\det{  Q^{-1} \frac{R_{1}}{R_{0}}} &= \frac{(R_{a,d}+R_{b,d})(R_{a,c}+R_{c,b})R^\text{eff}_1+
R_{a,d} R_{a,c} R_{b,c}+R_{a,c} R_{d,b} R_{b,c} + R_{a,d} R_{a,c} R_{d,b} + R_{a,d} R_{d,b} R_{b,c} }{R_{a,d} R_{c,b}- R_{a,c} R_{b,d}}\frac{1}{R_{0}}\frac{R_1}{R^\text{eff}_1} ,
\label{eqn:TransformationDeterminant2}
\end{align}
 where $R_0$ and $R_1$ are the true resistance in edges $e^0$ and $e^1$ while $R^\text{eff}_1$ is the effective resistance for $e^1$ from $\tilde{R}^\text{eff}$ and in $Q$. Now, we have derived equation \eqref{eqn:correlation}, we have,
\begin{eqnarray}
     i_1= \det Q \frac{R_0}{R_1} i_0  .
\end{eqnarray}
 
 The two-point effective resistance is superadditive, and when all the effective resistance values in our effective circuit are positive, $R^{\text{eff}}_{(i,j)}+R^{\text{eff}}_{(j,k)}\geq R^{\text{eff}}_{(i,k)}$, e.g., $R_{a,d}+R_{b,d}\geq R_0$. When $R^\text{eff}_1$ is much greater than the other effective resistance values, we can simplify equation \eqref{eqn:TransformationDeterminant1}, 
\begin{equation}
    \det{ Q\frac{R_{0}}{R_{1}}} \leq
    \frac{R_{a,d} R_{c,b}- R_{a,c} R_{b,d}}{R_0 R_1 }   .
\end{equation}
We expect $R^\text{eff}_1$ to be large when there are few parallel paths between $n_c$ and $n_d$, as when it is an input or output edge for an otherwise densely connected network.

It is experimentally more realizable to work with the two-point effective resistance for each of the six edges in our effective circuit, i.e., the diagonal elements in $\tilde{R}^\text{eff}$. It is possible to estimate the current correlations using the two-point effective resistances in equations \eqref{eqn:TransformationDeterminant1} and \eqref{eqn:TransformationDeterminant2}. We can determine the error in approximating the conductivity $L^\text{eff}_{a,b}$ by the two-point effective resistance. Given that $\tilde{R}^\text{eff}$ it is not singular, we have
\begin{eqnarray}
L^\text{eff} &=& B\left(\tilde{R}^\text{eff}\right)^{-1}B^t \nonumber\\
&=&  B\left(R_d+\underline{R}\right)^{-1}B^t \nonumber \\
&=&  B R_d^{-1} B^t -B R_d^{-1}\underline{R}\left(R_d+\underline{R}\right)^{-1}B^t
\end{eqnarray}
where $R_d$ and $\underline{R}$ are the diagonal and off-diagonal elements of $\tilde{R}^\text{eff}$. We arrive at the last line using the Woodbury matrix identity. We note, in general, the matrix norm of $\underline{R}\left(R_d+\underline{R}\right)^{-1}$ is not restricted to values less than $1$, and thus the contribution from the off-diagonal terms can be significant. It remains to be determined under what network conditions the matrix norm is less than $1$. The resistance for an edge $e_{a,b}$ can be written
\begin{eqnarray}
R_{a,b} &=& \left(L^\text{eff}_{a,b}\right)^{-1} \nonumber 
\\ &=& \frac{1}{\left(B R_d^{-1}B^t\right)_{a,b}-\left(B R_d^{-1}\underline{R}\left(R_d+\underline{R}\right)^{-1}B^t\right)_{a,b} } \nonumber 
\\ &=& \frac{1}{\left(B R_d^{-1}B^t\right)_{a,b}}+ \frac{ \left(B R_d^{-1}\underline{R}\left(R_d+\underline{R}\right)^{-1}B^t\right)_{a,b}}{ \left( B R_d^{-1}B^t\right)_{a,b} \left( \left(B R_d^{-1}B^t\right)_{a,b}-\left(B R_d^{-1}\underline{R}\left(R_d+\underline{R}\right)^{-1}B^t\right)_{a,b}  \right)}
\end{eqnarray}
the first term on the right-hand side is two-point effective resistance for $e_{a,b}$; thus, the second term on the right is the error in approximating $R_{a,b}$ with the two-point effective resistance. The percent error will be proportionate to the inverse of the two-point effective resistance, $\left(\left(B R_d^{-1}B^t\right)_{a,b}\right)^{-1}$; thus, we expect this approximation to be valid for networks with high resistance edges.

We can write the equations of motions of the memory parameter, $x$, under bias; we can calculate the divergence of the resistance values from an initial homogeneous resistance state given the full effective resistance. In the direct parameterization, we have,   
\begin{eqnarray}
    \dot{x}_a(t_0)-\dot{x}_b(t_0)=  \frac{\Ron}{\beta}\left( \det Q \frac{R_b}{R_a}-1\right){i}_b(t_0)  .
       \end{eqnarray}

Similarly we can estimate the trajectory of $\Delta x_{a,b}$ as a function of time:

\begin{eqnarray}
    x_a(t)-x_b(t)
   = \int^t_{t_0} -\alpha \left(x_a(s)-x_b(s)\right) +\frac{\Ron}{\beta}
        \left( \det Q(s)\frac{R_b(s)}{R_a(s)}-1\right){i}_b(s)     ds   .
\end{eqnarray} 
In general, the appropriate dynamics of $R$, $Q$, and $\vec{i}$ will depend on the memory parameters $x$, and thus solving for the trajectory $\Delta x_{a,b}$ must be done in a self-consistent way.
This approach is well suited for passive resistant or memristive circuits, wherein there are no sources between the edges of interest as the contributions of voltage or current sources are not accounted for in this current correlation. This is the case within the bulk of most memristive device networks.

\section{Discussion}
In conclusion, the work presented here addresses many fundamental questions about the performance of memrisitve networks. The dynamics of memristor networks are nontrivial, and relating these dynamics to network structure is even more challenging. The dynamics of arbitrary memristor circuits with other two-terminal circuit elements are even more challenging to understand but are important as they produce even richer dynamics. In general it is not known whether a network will have stable dynamics, and when different circuits will have similar dynamics. In the work presented above we demonstrate techniques for answering many of these questions for a rich class of memristor networks.

Our research presents a comprehensive framework for analyzing meristor circuits, offering valuable insights into their dynamic behavior. We introduced techniques such as direct and flipped parameterization, extending to encompass resistance variations dependent on polynomial expansions of the internal memory parameters, $x$, through Bernstein polynomials.
Moreover, we elucidated diverse control schemes to deduce equations of motion, providing a thorough understanding of network dynamics.  We generalize the non-orthogonal projection operators to identify the equations of motion of circuits with non-linear two-terminal circuit elements.
 From these equations, the eigenvalues give insight into the network properties; in particular, the RLC networks display resonator behavior in the under-dampened case, and memristors enhance system dampening.

From the equations of motions for $x$, we derived Lyapunov functions and demonstrate that for memristors linear in $x$ the dynamics are passive and the equilibrium states are stable. Importantly this is not guaranteed for memristors nonlinear in $x$. In the case where $g(X)=\tanh(X)$, e.g., a suitable activation function in neural networks, we could guarantee the memristive device network would have passive dynamics. The Lyapunov functions provide a powerful way to study memristive device networks' complicated and often chaotic dynamics.

Additionally, our study delved into the role of invariances, shedding light on the extensive gauge freedom present in circuit components orthogonal to projector matrices. We found gauge transformations do not alter the short time dynamics or Lyapunov function stability for linear memristors. Spectral invariances implemented through an orthogonal transformation had similar passive Lyapunov functions, however they significantly impact the time derivative of these functions, highlighting the complex interplay between network properties and orthogonal components.

Moreover, we showed it is possible to devise an effective circuit to study the correlations across a network. We established a direct correspondence between the graph Laplacian and the non-orthogonal node projection operators, facilitating a deeper understanding of network properties.
Leveraging these techniques, we constructed effective circuits between distant edges, from which it is possible to estimate the divergence of the internal memory parameter in distant edges. This provides a method to understand how the resistance evolves across a network.

These techniques can have direct application in experiments wherein the full network conductivity cannot be accessed. Experiments are often restricted to four point probe measurements that characterize the effective resistance between contacts. Under these conditions, effective circuits can serve as a powerful tool to relate the electrical properties between distant edges. The effective circuits we presented can be generalized to study the relation between input and output edges in a neuromorphic network.

These results demonstrate the rich dynamics within memristive device networks. The strength of network analysis is that it can obtain many useful governing equations and relations in a complicated circuit. Further work is needed to generalize these Lyapunov functions and correlation analysis to networks with capacitors, inductors, sources, and sinks. Such a generalization would provide a set of powerful techniques to describe memristive device networks in different complex network structures. This approach offers valuable insights into resistance evolution within networks, paving the way for advancements in circuit analysis and design.


\section{Acknowledgements}
Ongoing work by Samip Karki identified the Schottky barrier parameterization. The work of F. Caravelli and F. Barrows was carried out under the auspices of the NNSA of the U.S. DoE at LANL under Contract No. DE-AC52-06NA25396. F. Caravelli was financed via DOE LDRD grant 20240245ER, while F. Barrows via Director's Fellowship. F. Barrows gratefully acknowledges support from the Center for Nonlinear Studies at LANL.
\section{Citations}
\bibliography{bibliography}

\section{Supplementary}

\subsection{Space of Circuit Variables}

Here, we prove that current and voltage configurations are dual and that there is a correspondence between the two representations of a circuit, the cycle space corresponding to the loops of the cycle matrices $A$, and the vertex space corresponding to the vertices of the incidence matrix $B$. 
We define four-spaces
\begin{itemize}
\item {\bf CURR-SP} (current space): The set of all current configurations.  (i.e. ${\cal N}(B)$, the null-space of $B$)
\item {\bf VOLT-SP} (voltage space): The set of all voltage configurations.  (i.e. ${\cal N}(A)$, the null-space of $A$)
\item {\bf CYCLE-SP} (cycle space): The set of all linear combinations of cycles.  (i.e. ${\cal R}(A)$, the row-space of $A$)
\item {\bf VERT-SP} (vertex space): The set of all linear combinations of rows of B.  (i.e., ${\cal R}(B)$, the row-space of $B$)
\end{itemize}
We note that every row of $A$ is a valid current configuration: for every vertex $k$ (row of $B$) and cycle $l$ (row of $A$), if $k\notin l$ then it does not contribute, and otherwise, one edge of the cycle is directed towards  $k$ while another is directed away, canceling. Thus $BA^t = 0$ and as a consequence ${\cal{R}}(B)\subset {\cal{N}}(A)$, thus:
\begin{align}
\text{{\bf CYCLE-SP} }\subset \text{{\bf CURR-SP}}.
\end{align}

Taking the transpose of the above matrix equation, $AB^t = 0$ and as a consequence ${\cal{R}}(A)\subset {\cal{N}}(B)$, thus:
\begin{align}
\text{{\bf VERT-SP} }\subset \text{{\bf VOLT-SP}}.
\end{align}
The next step will be to show that the opposite inclusions hold, and any current (voltage) may be written $A^t y$ ($B^t y$).

\subsubsection{Spanning Trees: $\text{{\bf VOLT-SP}} = \text{{\bf VERT-SP}}$}

Assume the graph is connected with $n$ vertices.  As such, $m \ge n-1$.  A \emph{tree} is a connected graph with no cycles. A \emph{spanning tree} of $\mathcal{G}$ contains all vertices $T = (V, E_T)$ where $E_T$ are the \emph{tree branches}.  The remaining edges are chords. Every spanning tree has $n-1$ branches.

Consider a spanning tree of a graph $\mathcal{G}$ and choose one node of this tree as the \emph{root} or \emph{ground node}. There is a unique path from the root to any node $k$. Define the orientation along this path from the root to $k$. Now, for a given voltage configuration, define $p_k$ as the sum of voltage elements along the path with the appropriate orientation (added if the orientation of the tree and $\mathcal{G}$ agree and subtracted if not). With this construction, for any edge $(k, l)$ the voltage is $v_e = p_k - p_l$ and we can write,
\begin{align}
v = B^t p, \quad \text{{\bf VOLT-SP}} \subset \text{{\bf VERT-SP}}
\end{align}

And we have thus shown $\text{{\bf VOLT-SP}} = \text{{\bf VERT-SP}}.$  As every voltage can be defined by $n-1$ numbers in the potential, we have
\begin{align}
\text{dim}\; \text{{\bf VOLT-SP}} = n-1
\end{align}

The relationship $v = B^t p$ is also analogous to writing a voltage configuration as a sum of Green's functions given by the rows of $B$ as each row corresponds to the potential configuration on the nodes, $p_k = \delta_{ki}$ such that the potential is 1 for node $i$ and zero elsewhere.

\subsubsection{Algebraic Methods: $\text{{\bf CURR-SP}} = \text{{\bf CYCLE-SP}}$}

As a consequence of the above argument,
\begin{align}
i^t v = i^t B^t p = (B i)p = 0
\end{align}
which is known as Tellegen's theorem. So, the current configurations live in an orthogonal space to the voltage configurations. We need to do a bit of work to make use of this.

First, pick a maximal set of the voltage configurations $v_1\dots v_{n-1}$ by using, for example, $n-1$ rows of $B$. Now pick a maximal set of linearly independent columns of $A$, $a_1\dots a_r$ such that $a_i = Au_i$ where $u_i$ is a unit vector. We claim, $u_1\dots u_r, v_1, \dots v_{n-1}$ are linearly independent.  Suppose there are $\lambda_i, \mu_j$ such that,
\begin{align}
\sum_i \lambda_i u_i + \sum_j \mu_j v_j = 0.
\end{align}
Acting with $A$, as $Av=0$, we have $\lambda_i = 0$ by the linear independence of $u_i$ and then $\mu_i=0$ by the linear independence of $v_j$. For an arbitrary $m$ vector $z$, we can write
\begin{align}A z = \sum_i \lambda _i a_i = Ay, \quad y = \sum_i \lambda_i u_i\end{align}
where we have used the linearly independent set of columns and unit vectors from above. Thus writing $z = (z-y) + y$ we decompose the vector in a piece $z-y \in \text{{\bf VOLT-SP}}$ and another which can be written in terms of $r$ unit vectors, proving,
\begin{align}m = n - 1 + r.\end{align}
We can thus form a linearly independent set of the $n-1$ voltage vectors and $r$ linearly independent rows of A, $v_1\dots v_{n-1},c_1\dots c_r$ with which we can decompose an arbitrary $m$ vector.

So, writing $i = c + v$ as a shorthand for the decomposition, we must have $i^t v = 0 = c^t v + v^t v.$ The first term must vanish as $c^t v = c^t B p = (B^t c)p = 0$ and thus $v = 0$. So, $i$ may be written as a linear combination of the rows of $A$, and we have shown the opposite inclusion, giving
\begin{align}\text{{\bf CURR-SP}} = \text{{\bf CYCLE-SP}}.\end{align}
Additionally, we have
\begin{align}\text{dim}\; \text{{\bf CURR-SP}} = m - n+1\end{align}

\subsection{Differential equation for the resistance}
We note that it is possible to write a differential equation for the physically accessible resistive state without referring to the internal parameters $ x$. In the case in which $R(x)$ is linear, this is rather straightforward and we have
\begin{equation}
    \frac{d}{dt} \vec R(t)=\tilde R \left(\alpha \frac{\vec R-\Ron}{\tilde R}-\frac{1}{\beta } \left(I+\xi\Omega  \frac{R-\Ron}{\tilde R}\right)^{-1} \Omega \vec S\right)
\end{equation}
where $\tilde R =\Roff-\Ron$. 

On the other hand, if $R(x)$ is not linear the equation is slightly more involved. We use
\begin{equation}
    \frac{d}{dt} R\left(x(t)\right)=\partial_x R(x) \frac{d}{dt} x(t)
\end{equation}
and thus
    \begin{equation}
    |\partial_x R(x)|^{-1}\frac{d}{dt} R\left(x(t)\right)= \frac{d}{dt} x(t)
\end{equation}
To write the equation as a function $R(t)$, we must invoke the convexity of $R(x)$ in $x$.
Let us define $R=f(x)$, and $x=f^{-1}(R)$. If $f^{-1}$ exists, e.g., if $f$ is an invertible function, then we can write the differential equation in terms of $R$ only. The function $f$ is invertible if, for instance, that $R(x)$ is a monotonic function on $R:[0,1]\rightarrow [\Ron,\Roff]$ as we expect in passive memristive devices.
We now use the theorem of the inverse function. Let us call $R=f(x)$, then it is not hard to see that since we assume that $f(x)$ is monotonic in $x$, we have
\begin{equation}
    \frac{1}{\partial_x R(x)}=\partial_R f^{-1}\left(R\right).
\end{equation}
We can thus write also in the nonlinear case, given a certain function $g(R)$
\begin{equation}
    \frac{d}{dt} \vec R(t)=g(R)\left(\alpha \vec G(R)-\frac{1}{\beta} (I+\xi \Omega G(R))^{-1} \Omega \vec S \right)
\end{equation}
where $G$, $g$, and $f$ are generic functions, with the condition that $g=\partial_R G$ and $G:[\Ron,\Roff]\rightarrow [0,1]$. 

\subsection{Nonlinear Lyapunov functions differential equations}
The differential equation given by the quadratic and external field terms, equations \eqref{eqn:QuadraticTerm} and \eqref{eqn:ExternalFieldTerm}, are linear nonhomogeneous differential equations of the first kind, of the type
\begin{equation}
    y'+p(z)y=f(z)
\end{equation}
with  $f(z)=\frac{G_n(z)}{z}$ in both cases, while $p(z)=\frac{2}{z}$ for the quadratic term and $p(z)=\frac{1}{z}$ for the external field term. In this case, by writing $y=u v$ we solve separately for the two equations
\begin{eqnarray}
v'(z)+p(z) v(z)&=&0  \\
u'(z)&=&\frac{f(z)}{v(z)}
\end{eqnarray}
The first equation is solved via
\begin{equation}
    \frac{d}{dz} \log v(z)=- p(z)\rightarrow v(z)=v(1) e^{-\int^z p(q) dq}
\end{equation}
and since $p(q)=\frac{k}{q}$, with $k=1,2$ we have $v(z)=\frac{v(1)}{z^2}$ for the function $F$ and $v(z)=\frac{v(1)}{z}$ for the external field.
Thus we have
\begin{equation}
    u'(z)= \frac{z^k}{v(1)} f(z),
\end{equation}
whose solution is 
\begin{equation}
    u(z)=u_0+\frac{1}{v(1)}\int^z  q^k f(q) dq.
\end{equation}

These are used to arrive at the forms of $F(z)$ and $Q(z)$ used in the main text.

\subsection{Explicit expressions for non-orthogonal projectors}
We would like to derive the properties of the non-orthogonal projector 
$\Omega_{B/R^{-1}}=B^t (B R^{-1} B^t) ^{-1} B R^{-1}$ in terms of the orthogonal projector operator $\Omega_B=B^t (B B^t) ^{-1} B$.
Let us write
\begin{eqnarray}
    \Omega_{B/R^{-1}}R&=&B^t (B R^{-1} B^t) ^{-1} B \nonumber \\
    &=&R_{\text{on}} B^t \left(B \left(\frac{R}{R_{\text{on}}}\right)^{-1} B^t\right) ^{-1} B,
\end{eqnarray}
let us write $\tilde R=\frac{R}{R_{\text{on}}}=I+\xi G(\vec X)$. We have
\begin{eqnarray}
    \Omega_{B/R^{-1}}R&=&R_{\text{on}}B^t\left(B(I+\xi G(X))^{-1}B^t \right)^{-1}B \nonumber \\
    &=&R_{\text{on}}B^t\left(\sum_{k=0}^\infty (-1)^k \xi^k Bg^k B^t \right)^{-1}B \nonumber \\
    &=&R_{\text{on}}B^t\left(B B^t+B ZB^t\right)^{-1} B
\end{eqnarray}
where $Z=\sum_{k=1}^\infty (-1)^k \xi^k G^k$, as $\xi$ is a constant.
We now use the formula $(P+Q)^{-1}=\sum_{k=0}^\infty (-1)^k (P^{-1} Q)^k P^{-1}$.
We obtain, for $P=B B^t$ and $Q=B ZB^t$,
\begin{eqnarray}
    \Omega_{B/R^{-1}}R&=&R_{\text{on}}\sum_{k=0}^\infty (-1)^k B^t\left((B B^t)^{-1} B Z B^t)^k (B B^t)^{-1} B\right) \nonumber \\
    &=&R_{\text{on}}\sum_{k=0}^\infty (-1)^k \left(B^t(B B^t)^{-1} B Z \right)^k B^t(B B^t)^{-1} B \nonumber \\
    &=&R_{\text{on}}\sum_{k=0}^\infty (-1)^k (\Omega_B Z)^{k} \Omega_B \nonumber \\
    &=& R_{\text{on}}(I+\Omega_B Z)^{-1} \Omega_B
\end{eqnarray}
and thus
\begin{eqnarray}
    \Omega_{B/R^{-1}}&=&(I+\Omega_B (\tilde R^{-1}-I))^{-1} \Omega_B \tilde R^{-1} \nonumber \\
    &=&(\Omega_A+\Omega_B \tilde R^{-1})^{-1} \Omega_B \tilde R^{-1}
\end{eqnarray}

Here $\Omega_A$ is the orthogonal projection operator of $\Omega_B$ such that $\Omega_B+\Omega_A=I$. Let us prove that the equation defines a projector operator. 
Let us define for simplicity $\tilde R^{-1}-I=q$. In order to prove that the equation above is an identity, it is sufficient to observe that
\begin{equation}
    [\Omega_B (q+I)][(I+\Omega_B q)^{-1}]=\Omega_B.
\end{equation}
In fact, since $\Omega_B^2=\Omega_B$, we have
\begin{equation}
    \Omega_B (q+1)=\Omega_B+\Omega_B^2 q=\Omega_B+\Omega_B q,
\end{equation}
which proves the equality.
Then, we have
\begin{eqnarray}
\Omega_{B/R^{-1}}^2&=&(I+\Omega_B q)^{-1} \Omega_B (q+I)(I+\Omega_B q)^{-1} \Omega_B (q+I) \nonumber \\
&=&(I+\Omega_B q)^{-1} \Omega_B\Omega_B (q+I) \nonumber \\
&=&(I+\Omega_B q)^{-1} \Omega_B (q+I)=\Omega_{B/R^{-1}},
\end{eqnarray}
which shows that the formula we derived defines a projector operator.
Note that if $R=I$, then the formula becomes an orthogonal projector operator.
Let us now construct the orthogonal complement, which is defined as
\begin{equation}
    \Omega_{B/R^{-1}}+\Omega^\prime= I.
\end{equation}
We see that
\begin{eqnarray}
    \Omega^\prime&=& I-\Omega_{B/R^{-1}} \nonumber \\
    &=&(I+\Omega_B (\tilde R^{-1}-I))^{-1} (I+\Omega_B (\tilde R^{-1}-I)-\Omega_B \tilde R^{-1}) \nonumber \\
    &=&(I+\Omega_B (\tilde R^{-1}-I))^{-1} (I-\Omega_B) \nonumber \\
    &=&(\Omega_A+\Omega_B \tilde R^{-1})^{-1} \Omega_A.
\end{eqnarray}
In the case in which the projector is orthogonal, we also see in this case that for $\tilde R=I$, $\Omega^\prime=\Omega_A$ which is the right complementary projector to $\Omega_B$.
In the memristive RLC dynamics, we will need explicit expressions both for $\Omega_{B/R^{-1}}$ and $I-\Omega_{B/R^{-1}}$. 

\subsection{ Power in circuit cycles}

The matrix $K=A R A^t$ can be used to calculate the power within the cycles of a memristive circuit via a mesh analysis. In the mesh analysis, the memristive cycles form a planar sub-circuit that is partitioned into fundamental cycles, all with the same orientation. Cycles overlap in no more than one edge and each of these fundamental cycles is assigned a unique current. In this case, $P=\vec{i}^t_m K \vec{i}_m$, where $\vec{i}_m$ is the mesh current. In this case, the off-diagonal elements of $K$ will be negative and we can write the power in the mesh analysis as,
\begin{eqnarray}
P &=& \vec{i}^t_m K \vec{i}_m
\nonumber \\
&=& \sum_{c_i} \sum_k R^{c_i}_k i_{c_i}^2 - 2 \sum_{c_i < c_j} R^{c_i,c_j} i_{c_i}i_{c_j}
\nonumber \\
&=& \sum_{c_i} R^{c_i,c_i} i_{c_i}^2+\sum_{c_i< c_j} R^{c_i,c_j}(i_{c_i}-i_{c_j})^2 .
\end{eqnarray}
 
 Here, $i_{c_i}$ is the current in a cycle $c_i$ in the mesh analysis, $\sum_k R^{c_i}_k$ is the sum of all $k$ resistive elements in a cycle $c_i$, $R^{c_i,c_i}$ and $R^{c_i,c_j}$ are resistive elements that exist exclusively in $c_i$ or are shared in both $c_i$ and $c_j$, respectively. The off-diagonal elements in $K$ are the $R^{c_i,c_j}$ resistors.

 For example, in a simple mesh circuit consisting of two cycles with three resistors shown in Figure \ref{fig:2Mesh}, then
 \begin{eqnarray}
     \vec{i}_m &=& \begin{pmatrix}
         i_1 \\ i_2
     \end{pmatrix} 
     \\
     K &=& \begin{pmatrix}
         R_0+R_1 & -R_1 \\
         -R_1 & R_1+R_2
     \end{pmatrix} ,
 \end{eqnarray}
and the power in the cycles is
\begin{eqnarray}
    P &=& \vec{i}^t_m K \vec{i}_m \nonumber \\
    &=& R_0 i_1^2+R_2 i_2^2+R_1(i_1-i_2)^2 ,
\end{eqnarray}
which is the power we would expect from mesh currents.

\begin{figure}[h]
 \includegraphics[width=.4\textwidth]{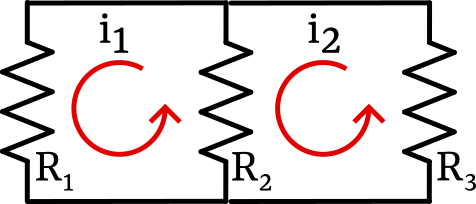}
 \caption{A small circuit consisting of two cycles with currents $i_1$ and $i_2$,and three resistors. The mesh analysis cycles are shown with red arrows.}
    \label{fig:2Mesh}
\end{figure}
 
\subsection{Effective power loss}
We can rewrite the power dissipation due to Joule heating in our effective circuit using $\tilde{R}^\text{eff}$ from the main text, 
\begin{eqnarray}
    P &=& \vec{i}^t\tilde{R}^\text{eff} \vec{i} \nonumber \\
    &=& \sum_{i} \vec{i}_i^2 \tilde{R}^\text{eff}_{i,i}+\sum_{i,j} \vec{i}_i\vec{i}_j \tilde{R}^\text{eff}_{i,j}  .
\end{eqnarray}
As $\tilde{R}^\text{eff}$ has off-diagonal terms, here we present a mean-field treatment to calculate the power dissipation when the current in each edge can be treated as nearly uniform with a local fluctuation, e.g., $\vec{i_i}=j_0+\delta \vec{i}_i$, $j_0$ is the mean current and $\delta \vec{i}_i$ is the small local fluctuation in current. We rewrite the power dissipation,
\begin{eqnarray}
    P &=& \sum_i j_0^2(\tilde{R}^\text{eff}_{ii}+\sum_j \tilde{R}^\text{eff}_{ij})+2j_0\delta\vec{i}_i(\tilde{R}^\text{eff}_{ii}+\sum_j \tilde{R}^\text{eff}_{i,j}) \nonumber \\ 
    &=& -\sum_i j_0^2(\tilde{R}^\text{eff}_{ii}+\sum_j \tilde{R}^\text{eff}_{ij})+2j_0\vec{i}_i(\tilde{R}^\text{eff}_{ii}+\sum_j \tilde{R}^\text{eff}_{i,j}) \nonumber \\
     &=& -\sum_i (j_0-\vec{i}_i)(j_0-\vec{i}_i)(\tilde{R}^\text{eff}_{ii}+\sum_j \tilde{R}^\text{eff}_{ij})+\vec{i}_i^2(\tilde{R}^\text{eff}_{ii}+\sum_j \tilde{R}^\text{eff}_{i,j}) \nonumber\\
     &=& \sum_i \vec{i}_i^2(\tilde{R}^\text{eff}_{ii}+\sum_j \tilde{R}^\text{eff}_{i,j}) .
\end{eqnarray}
In the first line we have neglected second-order fluctuation terms, and in the second line we have used $\delta \vec{i}_i=\vec{i}_i-j_0$. In the third line we have completed the square, and in the fourth line we again have neglected second-order fluctuations terms. Thus we can see the effective resistance $R_{i}=(\tilde{R}^\text{eff}_{ii}+\sum_j \tilde{R}^\text{eff}_{i,j})$, the sum of a row in $\tilde{R}^\text{eff}$. This approximation is valid only when the current fluctuations are small. 

\subsection{Correlation Analysis}
\subsubsection{Two-point effective resistance by a Markov process}
Before we find the effective resistance of a circuit, it is helpful to work with the nodes of the circuit. The movement of charge in a network can be studied via a Markov process, wherein current jumps between nodes with a probability depending on the conductivity of the edges. We build a Markov matrix $P$, the matrix elements encode the normalized probabilities that a charge moves between adjacent nodes; $P(a,b)$ is the probability of moving from node $n_a$ to node $n_b$. 
\begin{eqnarray}
    P(a,b) &=& \frac{G_{ab}}{\sum_c G_{ac}}
    \nonumber\\ &=& \bar{G}_{a,b}
\label{eqn:MarkovMatrix}
\end{eqnarray}
Here $G_{ab}$ is the conductivity of the edge $e_{ab}$, and the sum over $c$ is the sum over all edges incident on $n_a$. $\bar{G}_{a,b}$ is the conductivity of $e_{ab}$ normalized by the conductivity of all edges incident on $n_a$. The matrix $P$ has a zero diagonal.

We gain insight into the current in a memristive device network by studying the paths charged particles can take between nodes. The probability of walks in this network (paths that the current could take) can be calculated from $P$. The total probability of all $k$-step long walks, $p_w$, that starts at $n_a$ and reaches a distant node $n_b$ can be calculated,
\begin{eqnarray}
p_w= P^{k}_{a,b}  .
\end{eqnarray}
As the charge moves through the network in a fast time scale compared to the change in the memristive device, we need to examine long walks. Given the values of $P(a,b)$ are small we can write the sum of all the possible walks in the network, $\sum^\infty P^k$, as:
\begin{eqnarray}
    \left(I-P\right)^{-1} &=& \frac{1}{\bar{L}}
    \nonumber\\
    &=& \bar{L}^+
\end{eqnarray}
Here, $\bar{L}$ is the normalized Laplacian matrix. In the case of homogenous resistors, the weights are equal to the inverse degree of each edge, and in this case, $\bar{L}$ reduces to the standard normalized Laplacian matrix. In the more general case, the probability is the normalized conductivity of all edges connected to a given node, and the diagonal elements are $1$, as given in 
\begin{eqnarray}
\bar{L}= \text{diag}(B G B^t)\cdot(B G B^t) ,   
\label{eqn:Laplacioan_normed}
\end{eqnarray} 
here $\text{diag}(BGB^t)$ is the diagonal elements of $BGB^t$, which normalizes the conductivity. 
We arrive at equation \eqref{eqn:Laplacioan_normed} using Theorem \ref{thm:Laplacian} and dividing each row by the corresponding diagonal elements of the matrix such that the off-diagonal elements of each row sum to $-1$ and each row sums to $0$.

If we define
$\bar{L}^+$ to be the pseudoinverse of $\bar{L}$, then we have
\begin{eqnarray}
\bar{L}^+=(B G B^t)^{-1}\cdot \text{diag}(B G B^t).
\label{eqn:Laplacioan_normed_inverse}
\end{eqnarray}

The effective resistance is a resistance such that if the entire network was removed, leaving only nodes $n_a$ and $n_b$ with an edge $(a,b)$ with $R^\text{eff}$, the potential difference and corresponding electrical flows between these nodes would be invariant. As such, we can find an effective circuit in order to calculate the correlation between electrical currents in non-adjacent edges. The effective resistance, $R^\text{eff}$, between the $n_a$ and $n_b$ can be found by examining the $(a,b)$ walks and removing the closed cycle walks which return to $n_a$, $(a,a)$, and similarly examining the reverse walk from $(b,a)$:
\begin{eqnarray}
    p_{a,b} &=& -\bar{L}^+(a,b)-\bar{L}^+(b,a)+\bar{L}^+(a,a)+\bar{L}^+(b,b)
    \\
    R^{\text{eff}}_{(a,b)}(a,b) &=& -\frac{\bar{L}^+(a,b)}{(B G B^t)(b,b)}-\frac{\bar{L}^+(b,a)}{(B G B^t)(a,a)}+\frac{\bar{L}^+(a,a)}{(B G B^t) (a,a)}+\frac{(B G B^t)(b,b)}{(B G B^t)(b,b)} \nonumber
    \\
    &=& -L^+(a,b)-L^+(b,a)+L^+(a,a)+L^+(b,b)
    \label{eqn:2point_effectiveresistance_supp}
\end{eqnarray}

Here, $L^+$ is the unnormalized inverse Laplacian; we have found the effective resistance by examining the walks in $\bar{L}^+$ and dividing by the normalization factor to restore the resistance values.

\subsubsection{Current Correlation in a Random Walk}

 Here we examine a toy model wherein current moves through the memristive device network via a Markov process\cite{DoyleSnell_Arxiv_2000}. As we are interested in studying the correlation of currents in edges, we build a Markov matrix, $P$, wherein charge jumps between edges. In effect, we transform the network into the line graph wherein edges in our original network are transformed into nodes. Edges that are adjacent in our network become adjacent nodes in our new graph; these new nodes are connected by asymmetric edges.
 A Markov matrix, $P$, encodes the probability of a stochastic process, e.g., a particle flowing through the network randomly without external bias, diffusing in the line graph. The probability depends on the conductivity of the edges:
\begin{eqnarray}
    P(a,b) &=& \frac{G_{ab}}{\sum_c G_{ac}}
    \nonumber \\ &=& \bar{G}_{a,b}
\label{eqn:MarkovMatrix2}
\end{eqnarray}
Here $G_{ab}$ is the conductivity of edge $e^b$ in our original network, which is adjacent to edge $e^a$, and $G_{ba}$ is the conductivity of edge $e^a$ in our original network (adjacent to $e^b$). The sum over $c$ in equation \eqref{eqn:MarkovMatrix2} is over all edges $e^c$ adjacent to $e^a$ in our original network.  $\bar{G}_{a,b}$ is the conductivity of $e^b$ adjacent to $e^a$ normalized by the conductivity of all edges adjacent to $e^a$. We can see that the matrix $P$ has a zero diagonal as above, but now $P$ is asymmetric.

As our random walk occurs via all-or-nothing jumps, we set a bound on the correlation of current in our network by underestimating the probability that current will propagate through the network. In our model, a particle can move in either direction on a given edge without accounting for the conductivity of the edge, i.e., a particle is equally likely to hop from either end of a given edge. Thus we overestimate the probability the charge will, in effect, be reflected by an edge.

The probability $h_{(a,b)}$ is the probability a particle walks from edge $e^a$ and travels to a distant edge $e^b$ along the shortest path; now $e^b$ is not necessarily directly adjacent to $e^a$. The random walk probability $h_{(a,b)}$ can be calculated from $P$,
\begin{eqnarray}
    h_{(a,b)}&=&\sum_{k=1} P^k\,_{a,b} \; \vert \; \sum P^{k-1} \, _{a,b}=0  .
\end{eqnarray}

In a network, charge will not flow in cycles but instead flow in a directed manner until reaching an equilibrium. Thus, we are only interested in finding the shortest path when the charge does not return to any edge it has traversed. For example, in a simple square circuit with homogeneous resistance:
\begin{eqnarray}
    P=\begin{pmatrix}
        0 & \frac{1}{2} & 0 & \frac{1}{2} \\
        \frac{1}{2} & 0 & \frac{1}{2} &0\\
        0 & \frac{1}{2} & 0 & \frac{1}{2} \\
         \frac{1}{2} & 0 & \frac{1}{2} &0\\
    \end{pmatrix}
\end{eqnarray}
in this case $h_{(a,b)}=2 \frac{1}{2}^2 $. This random walk probability, $h_{(a,b)}$ is a minimum probability a particle injected in $e^a$ will travel through $e^b$. This can be seen by noting that for a unit charge injected in $e^{a}$, then $P(a,b)$ is the probability the charge will flow through an adjacent $e^b$. $P(b,c)$ is then the probability that charge will move from $e^b$ to an adjacent edge $e^c$. Note that in this treatment, a charge is unbiased in the direction in which it flows; a charge injected into any edge will move to adjacent edges based solely on the conductivity of the adjacent edges.  

 In the systems we study, we treat the current as quasistatic as it equilibrates between edges and traverses the network at a fast time scale compared to the change in the resistance in the memristive devices. Thus, $h_{(a,b)}$ is a correlation of current in edge $e^a$ and $e^b$ under random walk conditions, e.g., when the current flow is unbiased. Given a known current ${i}_a$ in $e^a$, we can get a bound on the current  ${i}_b$ in $e^b$. We have
\begin{eqnarray}
    h(a,b)i_a\leq i_b\leq \frac{i_a}{h(b,a)}
\end{eqnarray}
For example, if $h_{(a,b)}=h_{(b,a)}=1$ the current is equivalent in edges $e^a$ and $e^b$, ${i}_a={i}_b$, if $h_{(a,b)}=h_{(b,a)}=0$ the current in the two edges are independent.  
The upper bound is found from the reverse walk: $h(b,a)i_b\leq i_a$, as $P$ is asymmetric $h_{(a,b)}\neq h_{(a,b)}$. 

These are inequalities as probability of moving between adjacent edges is smaller in the random walk case than the probability of moving between adjacent edges in the directed walk case; in the random walk case, there is a probability of moving in either direction along an edge. In addition, here we are just examining the probability that current injected at a single edge $e^a$ transmits through $e^b$. Under normal conditions, multiple edges could contribute to the current of any individual edge.

This correlation analysis can offer insight into correlations of the resistance values in distant memristive devices. Given a network of memristive devices initialized with a homogenous resistance, $X(t_0)_{k,k}=x_0$, the rate at which the memory parameter $x$ diverges in distant edges depends on $h(a,b)$. In the direct parameterization, we have,
\begin{eqnarray}
  \frac{\Ron}{\beta} \left( h_{(a,b)}(t_0)-1\right){i}_b(t_0)  \leq  \dot{x}_a(t_0)-\dot{x}_b(t_0) 
        \leq \frac{\Ron}{\beta} \left( h_{(b,a)}(t_0)^{-1}-1\right){i}_b(t_0)   .
\end{eqnarray}
We can simply put bounds on the divergence of $\Delta x_{a,b}=x_a-x_b$. Similarly, we can put bounds on the trajectories of $\Delta x_{a,b}$, which relates the difference of resistance in distant memristive devices:
\begin{align}
  \int^t_{t_0} -\alpha \left(x_a(s)-x_b(s)\right) +\frac{\Ron}{\beta} \left( h_{(a,b)}(s)-1\right){i}_b(s) ds   \leq \Delta x_{a,b}(t) 
        \leq \int^t_{t_0} -\alpha \left(x_a(s)-x_b(s)\right) +\frac{\Ron}{\beta} \left( h_{(b,a)}(s)^{-1}-1\right){i}_b(s)  ds 
\end{align}

In general, the appropriate dynamics of $R$, $h$, and $\vec{i}$ will depend on the memory parameters $x$, and thus solving for the trajectory $\Delta x_{a,b}$ must be done in a self-consistent way. 
Inherent in the analysis thus far is the charged particles are traveling in a network nearly randomly; inhomogeneities in the current distributions are due solely to network connectivity. In this treatment, the voltage in the network does not strongly bias the propagation of electrons. This is only valid under low voltage bias such that the potential in the network can be considered locally flat; in effect, particles are being injected in an edge, but there is no bias in the network. In the main text, the case of a circuit under inhomogeneous applied bias in a network with spatially varying resistance is studied.

\subsubsection{Wheatstone bridges}

We have two Wheatstone bridge circuits, which are equivalent representations of the circuit built from six effective resistances, as described in the main text. In one circuit, we can treat the edge with a known current, $e^0$, as a current injector that connects the ends of the Wheatstone bridge and the unknown edge, $e^1$, is the bridge. In the other circuit, the unknown edge, $e^1$, is the source that connects the ends of the Wheatstone bridge, where the bridge is now $e^0$.

In order to solve for valid current configurations, we remove one edge of our circuit; in both cases, we remove the edge linking the ends of our Wheatstone bridge, the source edge. In the first case, we have $(V_a,V_b)$ at the top and bottom of our Wheatstone bridge, with $(V_c,V_d)$ on the ends of our bridge, as shown in Figure \ref{fig:effectiveResistance} (b) and (c) in the main text. We use an incidence matrix $B$ to solve for $(V_a,V_b)$ in terms of $(V_c,V_d)$. Similarly, in the second circuit the positions of $(V_c,V_b)$ and $(V_a,V_b)$ are reversed and we solve for $(V_c,V_b)$ in terms of $(V_a,V_b)$. The following linear mappings are found,

\begin{align}
    \begin{pmatrix}
        V_a \\ V_b
    \end{pmatrix}
    &=
    \begin{pmatrix}
        \frac{R_{1} R_{2} R_{3} + R_{1} R_{2} R_{4} + R_{1} R_{2} R_{5} + R_{1} R_{4} R_{5}}{R_{5} \left(R_{1} R_{4} - R_{2} R_{3}\right)} & \frac{- R_{1} R_{2} R_{3} - R_{1} R_{2} R_{4} - R_{1} R_{2} R_{5} - R_{2} R_{3} R_{5}}{R_{5} \left(R_{1} R_{4} - R_{2} R_{3}\right)}\\\frac{- R_{1} R_{3} R_{4} - R_{2} R_{3} R_{4} - R_{2} R_{3} R_{5} - R_{3} R_{4} R_{5}}{R_{5} \left(R_{1} R_{4} - R_{2} R_{3}\right)} & \frac{R_{1} R_{3} R_{4} + R_{1} R_{4} R_{5} + R_{2} R_{3} R_{4} + R_{3} R_{4} R_{5}}{R_{5} \left(R_{1} R_{4} - R_{2} R_{3}\right)}
    \end{pmatrix} \begin{pmatrix}
        V_c \\ V_d
    \end{pmatrix}
    \label{eqn:LinearTransformM}
    \\
    \begin{pmatrix}
    V_c \\ V_d    
    \end{pmatrix}
    &=
    \begin{pmatrix}\frac{R_{1} R_{3} R_{4} + R_{1} R_{4} R_{5} + R_{2} R_{3} R_{4} + R_{3} R_{4} R_{5}}{R_{1} R_{2} R_{3} + R_{1} R_{2} R_{4} + R_{1} R_{2} R_{5} + R_{1} R_{3} R_{4} + R_{1} R_{4} R_{5} + R_{2} R_{3} R_{4} + R_{2} R_{3} R_{5} + R_{3} R_{4} R_{5}} & \frac{R_{1} R_{2} R_{3} + R_{1} R_{2} R_{4} + R_{1} R_{2} R_{5} + R_{2} R_{3} R_{5}}{R_{1} R_{2} R_{3} + R_{1} R_{2} R_{4} + R_{1} R_{2} R_{5} + R_{1} R_{3} R_{4} + R_{1} R_{4} R_{5} + R_{2} R_{3} R_{4} + R_{2} R_{3} R_{5} + R_{3} R_{4} R_{5}}\\\frac{R_{1} R_{3} R_{4} + R_{2} R_{3} R_{4} + R_{2} R_{3} R_{5} + R_{3} R_{4} R_{5}}{R_{1} R_{2} R_{3} + R_{1} R_{2} R_{4} + R_{1} R_{2} R_{5} + R_{1} R_{3} R_{4} + R_{1} R_{4} R_{5} + R_{2} R_{3} R_{4} + R_{2} R_{3} R_{5} + R_{3} R_{4} R_{5}} & \frac{R_{1} R_{2} R_{3} + R_{1} R_{2} R_{4} + R_{1} R_{2} R_{5} + R_{1} R_{4} R_{5}}{R_{1} R_{2} R_{3} + R_{1} R_{2} R_{4} + R_{1} R_{2} R_{5} + R_{1} R_{3} R_{4} + R_{1} R_{4} R_{5} + R_{2} R_{3} R_{4} + R_{2} R_{3} R_{5} + R_{3} R_{4} R_{5}}\end{pmatrix}
    \begin{pmatrix}
        V_a \\ V_b
    \end{pmatrix}
    \label{eqn:LinearTransformQ}
\end{align}
The effective resistances here are $R_{a,b}=R_0$, $R_{a,d}=R_1$, $R_{a,c}=R_2$, $R_{b,d}=R_3$, $R_{b,c}=R_4$, and $R_{c,d}=R_5$. These are the linear transformations we use in the main text, with the matrices in equations \eqref{eqn:LinearTransformM} and \eqref{eqn:LinearTransformQ} corresponding to $M$ and $Q$, respectively. We note that $M=Q^{-1}$.

\subsubsection{Linear mapping between voltages}
Here we detail how the linear mappings between voltage in our unbalanced Wheatstone bridge circuit are used to determine a correlation between currents $i_0$ and $i_1$, in edges $e^0$ and $e^1$, with resistance $R_0$ and $R_1$, respectively. From the main text we have,
\begin{eqnarray}
   \begin{pmatrix}
        V_c \\ V_d
    \end{pmatrix}
    &=& 
    Q
    \begin{pmatrix}
    V_a \\ V_b     
    \end{pmatrix}
     \\
      \begin{pmatrix}
        V_a \\ V_b
    \end{pmatrix}
    &=& 
    Q^{-1}
    \begin{pmatrix}
    V_c \\ V_d     
    \end{pmatrix}  
\end{eqnarray}
We can use these mappings to calculate the currents from the resistance $R_0$ and $R_1$,
\begin{eqnarray}
    \begin{pmatrix}
        i_1 \\ 0
    \end{pmatrix} &=& \frac{1}{R_1} 
    \begin{pmatrix} 1 & -1 \\ 0 & 0
    \end{pmatrix} Q \begin{pmatrix} V_a \\ V_b \end{pmatrix}
   \nonumber \\
    &=& \frac{1}{R_1} 
    \begin{pmatrix} 1 & -1 \\ 0 & 0
    \end{pmatrix}  \begin{pmatrix} V_c \\ V_d \end{pmatrix}
    \\
    \begin{pmatrix}
        i_0 \\ 0
    \end{pmatrix} &=& \frac{1}{R_0} 
    \begin{pmatrix} 1 & -1 \\ 0 & 0
    \end{pmatrix} Q^{-1} \begin{pmatrix} V_c \\ V_d \end{pmatrix}
    \nonumber \\
     &=& \frac{1}{R_0} 
    \begin{pmatrix} 1 & -1 \\ 0 & 0
    \end{pmatrix}  \begin{pmatrix} V_a \\ V_b \end{pmatrix}    
\end{eqnarray}
 We can rearrange these matrices to find a mapping between electrical currents, which we use to find a ratio between $i_0$ and $i_1$, 
\begin{eqnarray}
   \begin{pmatrix}i_0 \\ 0 \end{pmatrix} &=& \frac{R_1}{R_0}\begin{pmatrix} 1 & -1 \\ 0 & 0
    \end{pmatrix}  Q^{-1} \begin{pmatrix} 1 & -1 \\ 0 & 0
    \end{pmatrix}^+\begin{pmatrix}i_1 \\ 0 \end{pmatrix}     
    \\
    i_1 &=& 2\frac{R_0}{R_1} \frac{Q(0,0)Q(1,1)-Q(0,1)Q(1,0)}{Q(0,0)+Q(0,1)+Q(1,0)+Q(1,1)} i_0 \nonumber
    \\ &=& \frac{R_0}{R_1} \det{(Q)} i_0
    \label{eqn:ChargeRatio1}
\end{eqnarray}
The superscript $^+$ is the pseudo-inverse. Note the denominator, $Q(0,0)+Q(0,1)+Q(1,0)+Q(1,1)$, cancels the factor of $2$. We can go through the same process to obtain the inverse mapping:
\begin{eqnarray}
   i_0 &=& \frac{R_1}{R_0} \det{(Q^{-1})} i_1 .
\label{eqn:ChargeRatio2}
\end{eqnarray}
If we are interested in the correlation between currents $i_0$ and $i_1$, we treat $i_0$ as a unit charge and equations \eqref{eqn:ChargeRatio1} and \eqref{eqn:ChargeRatio2} reproduce the current transformation from the main text. We can write,
\begin{eqnarray}
    i_0i_1 &=& \frac{R_1}{R_0}\det{(Q^{-1})}i_1^2
   \nonumber  \\
    &=& \frac{R_0}{R_1}\det{(Q)}i_0^2  ,
\end{eqnarray}
from which we define the dimensionless correlation $\langle i_1,i_0\rangle$ in the main text for unit current in $e^0$.

\end{document}